\documentclass[12pt,letter]{article}
\usepackage{amssymb,amsmath,amsthm,mathrsfs}
\usepackage[mathscr]{euscript}
\usepackage{graphicx}
\usepackage{setspace}
\usepackage{fullpage}
\usepackage{hyperref}
\hypersetup{
    colorlinks=true,
    linkcolor=black,
    filecolor=magenta,      
    urlcolor=cyan,
    pdftitle={Overleaf Example},
    pdfpagemode=FullScreen,
    }

\usepackage[makeroom]{cancel}
\usepackage{afterpage}
\usepackage{wasysym}
\usepackage[greek,english]{babel}
\usepackage{teubner}
\usepackage{titling}
\usepackage{blindtext}

\usepackage{setspace}
\usepackage{amsthm}
\usepackage[margin=3cm]{geometry}
\usepackage{amssymb,amsmath}
\usepackage{mathrsfs}
\usepackage{pstricks}

\usepackage{tikz-cd}
\usepackage[toc,page]{appendix}

\usepackage{blindtext}

\newcommand{\fourxn}{(x_n,y_n, w_n,z_n)}
\newcommand{\fourx}{(x,y,w,z)}
\newcommand{\fourX}{X \times X \times X \times X}
\newcommand{\ie}{\textit{i.e.}}
\newcommand{\wpp}{\succcurlyeq}
\newcommand{\R}{\mathbb{R}}

\newtheorem{theorem}{Theorem}
\newtheorem{axiom}{Axiom}
\newtheorem{proposition}{Proposition}
\newtheorem{lemma}{Lemma}

\newtheorem{mydef}{Definition}

\newtheorem{innercustomthm}{Axiom}
\newenvironment{customthm}[1]
  {\renewcommand\theinnercustomthm{#1}\innercustomthm}
  {\endinnercustomthm}

\usepackage[utf8x]{inputenc}
\usepackage{afterpage}
\makeatletter
\newcommand{\varnotsuccsim}{\mathrel{\mathpalette\varn@t\succsim}}
\newcommand{\varn@t}[2]{%
  \vphantom{/{#2}}%
  \ooalign{\hfil$\m@th#1/\mkern2mu$\cr\hfil$\m@th#1#2$\hfil\cr}%
}

\newcommand{\equivclass}[1]{%
  #1/{\sim}%
}

\usepackage{titlesec}

\titleformat*{\section}{\LARGE\bfseries}
\titleformat*{\subsection}{\Large\bfseries}
\titleformat*{\subsubsection}{\large\bfseries}
\titleformat*{\paragraph}{\large\bfseries}
\titleformat*{\subparagraph}{\large\bfseries}

\g@addto@macro\normalsize{%
  \setlength\abovedisplayskip{2.5pt}
  \setlength\belowdisplayskip{2.5pt}
  \setlength\abovedisplayshortskip{2.5pt}
  \setlength\belowdisplayshortskip{2.5pt}
}

\title{On the notion of measurable utility on a\\connected and separable topological space:\\an order isomorphism theorem.\thanks{I am grateful to Professor Massimo Marinacci for letting me know about the open problem.}​}
\author{Gianmarco Caldini}
\date{7 February 2020}

\begin{document}

\maketitle

\begin{abstract}
The aim of this article is to define a notion of cardinal utility function called measurable utility and to define it on a connected and separable subset of a weakly ordered topological space. The definition is equivalent to the ones given by Frisch in 1926 and by Shapley in 1975 and postulates axioms on a set of alternatives that allow both to ordinally rank alternatives and to compare their utility differences. After a brief review of the philosophy of utilitarianism and the history of utility theory, the paper introduces the mathematical framework to represent intensity comparisons of utility and proves a list of topological lemmas that will be used in the main result. Finally, the article states and proves a representation theorem, see Theorem \ref{thm1}, for a measurable utility function defined on a connected and separable subset of a weakly ordered topological space equipped with another weak order on its cartesian product. Under some assumptions on the order relations, Theorem \ref{thm1} proves existence and uniqueness, up to positive affine transformations, of an order isomorphism with the real line.​
\end{abstract}

\newcommand{\thicktilde}[1]{\mathbf{\tilde{\text{$#1$}}}}
\newcommand{\Aone}{\mathbf{A1}}
\newcommand{\Atwo}{\mathbf{A2}}
\newcommand{\Athree}{\mathbf{A3}}
\newcommand{\Aoneprime}{\mathbf{A1'}}
\newcommand{\Atwow}{\mathbf{A2w}}

\tikzcdset{every label/.append style = {font = \Large}}

\pagenumbering{arabic}

\addtocontents{toc}{\protect\enlargethispage{\baselineskip}}


\newpage

\section*{Introduction}
\addcontentsline{toc}{section}{Introduction}

\setcounter{page}{1}
Together with notions such as value, money, market and economic agents, utility has been one of the most controversial concepts in the whole history of economic theory. The most important debate can be considered the one around the question whether it is possible to define a clear and rigorous concept of utility and an appropriate notion of unit of measurement for utility, seen as a quantity like the physical ones. In the first chapter we will give a short introduction to the evolution of the concept of utility from both a philosophical and a historical point of view. Our treatment is far from being exhaustive. For an extensive treatment of history of utility and utility measurements, we refer the interested reader to Stigler [\ref{stigler50}], Majumdar [\ref{majumdar58}], Adams [\ref{adams60}], Luce and Suppes [\ref{lucesuppes65}], Fishburn [\ref{fishburn68}], [\ref{fishburn70}], [\ref{fishburn76}] and Moscati [\ref{moscati19}].

The second chapter will shift from the descriptive part to more formal concepts and will be used to introduce the usual mathematical framework of decision theory. Moreover, we will introduce definitions and axioms that will enable us to represent comparisons of the intensity that a decision maker feels about the desirability of different alternatives. For this aim, we will follow the construction of Suppes and Winet [\ref{suppeswinet55}] and Shapley [\ref{shapley75}].

The third and last chapter will be entirely devoted to the proof of Shapley's theorem, extending the domain of alternatives $X$ from a convex subset of $\R$ to a connected and separable subset of a topological space, hence providing a generalization of his theorem. Our intended goal is to define a rigorous notion of a specific kind of cardinal utility function, not only able to rank alternatives, but also to compare utility differences. In particular, we define a ``twofold" utility function in line with the primordial axiomatization of Frisch [\ref{frisch26}], calling it a $\textit{measurable}$ utility function. In mathematical terms, we will prove a specific order-isomorphism theorem between a totally ordered, connected and separable subset of a topological space and the real line.

\section{Philosophy and history of utility theory}

Theory of felicity, theory of justice, theory of morality, theory of virtue and theory of utility are among the most important theories of moral philosophy and, as such, they are constantly sources of questions that often do not find an immediate answer. When a human being acts, or when she makes a decision, she is, at the same time, looking for justifications, either positive or normative, for the decision she has just made. We, as human beings, are constantly trying to prove that what we did was the $\textit{best}$ thing to do, in some well-defined sense, or, at least, the less harmful. These justifications take into account the means, the ends and all the possible paths we have to reach our goals. Moral philosophy is the science that comes into place when we formulate questions about the ends, the means and the possible ways to achieve them. 

Moral philosophy is essentially composed by principles, also called $\textit{norms}$, on what is good and what is bad. They allow to define and to judge human actions, means and ends. Sometimes, norms take the form of $\textit{universal laws}$ to which all human beings are subjected. Nevertheless, the formulation of moral laws or rules that prescribe what a single agent should do or not do are intrinsically tied with history. Historical experiences determine our vision of the world. Our moral philosophy is the result of different heritages that formed a common culture in which values like human respect, an idea of equality between human beings and impartiality are among the most important.

Together with this general definition of morality, there exist the similar concepts of $\emph{ethics}$ and of $\emph{role morality}$ - a specific form of professional morality. It was Jeremy Bentham, in an unfinished manuscript which was posthumously published in 1834, to define the neologism $\textit{deontology}$ in the title of his book $\textit{Deontology or the Science of Morality}$. The manuscript stated, for the first and only time, the particular aspects of Bentham's utilitarian theory as moral philosophy. This passage is clearly mentioned in Sørensen [\ref{philosophy}]:

\begin{quote} \footnotesize
[...] pointing out to each man on each occasion what course of conduct promises to be in the highest degree conducive to his happiness: to his own happiness, first and last; to the happiness of others, no farther than in so far as his happiness is promoted by promoting theirs, than his interest coincides with theirs (p. 5). 
\end{quote}

In this passage we can see how Bentham considered deontology to be primarily aimed at one's own private felicity. Nevertheless, this does not bring any selfish concern. Bentham's goal can be identified with the objective study and measurement of passions and feelings, pleasures and pains, will and action. Among these particular pleasures are those stemming from sympathy - in Adam Smith's sense - and they include the genuine pleasure being happy for the good of others. 

In this light, Bentham spent his life in search of the cardinal principle of ethics and he found it in Epicurean ethics of hedonism. Hedonism comes from Greek $\textit{\<{η}δον\'{η}}$, which means $\textit{pleasure}$. Thus, classic utilitarianism, founded on hedonism, started from the principle that pleasure is an intrinsic positive value and sorrow is an intrinsic negative value. It is, for this reason, somehow curious that Bentham conception, founded on pleasure, had been called $\textit{utilitarianism}$, from the simple observation that what is $\textit{useful}$ is not necessarily pleasant or providing pleasure. We need always to take into account that the term $\textit{utility}$ is intended in a functional sense; what gives $\textit{utility}$ is what contributes the most to the individual, or universal, pleasure.

Classical utilitarian philosophers considered utilitarianism $\textit{well-founded}$ and $\textit{realistic}$ thanks to the fact that it is based on pleasure. It is well-founded as its norms are justified by an intrinsic, absolute value, that does not need any further justification. It is realistic because they thought human being to ultimately seek the maximum pleasure and the minimum sorrow. More specifically, human beings try to choose the action that will provide the maximum excess of pleasure against grief. 

For Bentham, what really matters is the total amount of pleasure, intended as the total excess of pleasure against sorrow: the only reasons for human actions are the quests for pleasure, avoiding sorrow: they are the sources of our ideas, our judgments and our determinations. Human moral judgments become statements on happiness; pleasure (or felicity) is good and sorrow is bad. Utilitarian moral can be considered as a ``calculated hedonism", that carefully evaluates the characteristics of pleasure. Wise is the man that is able to restrain from an immediate pleasure for a future good that, in comparison, will be more beneficial. On the other side, being able to evaluate the positive or negative consequences of an action without making mistakes is fundamental. Hence, the correct utilitarian person should reach some kind of ``moral arithmetic" that allows the correct calculations to be carried out. Far from being a unanimously accepted doctrine, we cannot forget to mention that Alessandro Manzoni wrote an essay [\ref{manzoni19}] in which he strongly criticized Bentham's utilitarianism, saying that it is utterly wrong to think that human beings build their moral values judgment of their actions on utility.\footnote{Manzoni [\ref{manzoni19}] wrote: \emph{"Non ci vuol molto a scoprir qui un falso ragionamento fondato sull’alterazione d’un fatto. Altro è che l’utilità sia un motivo, cioè uno de’ motivi per cui gli uomini si determinano nella scelta dell’azioni, altro è che sia, per tutti gli uomini, il motivo per eccellenza, l’unico motivo delle loro determinazioni (p.775).}}

From this explanation of utilitarianism, Bentham's evaluation criterion of actions follows as an immediate corollary: the maximum happiness for the maximum number of people. Again, happiness is intended as $\textit{state of pleasure}$, or $\textit{absence of grief}$. Hence, individual pleasure becomes no more the ultimate goal: it is the universal pleasure to be hegemonic.\footnote{This tension between individual pleasure and universal pleasure is one of the principal difficulties of utilitarian moral philosophy.} 

This view of utilitarianism admits, at least in our minds, the conception of the existence of a $\textit{scale}$ of pleasure in which pleasure and sorrow can be added and subtracted. In other words, the idea of a $\textit{calculus}$ of felicity and grief is not completely absurd, both in intrapersonal and interpersonal compensations.

\subsection{Brief history of utilitarianism}

Although it is possible to find utilitarian reasonings in Aristotele's works, it is commonly agreed that the beginning of the history of utility can be identified with 18th century moral philosophy. To be even more specific, Bentham's ideas were not isolated, since they were already present in works by his illuministic predecessors like Richard Cumberland, Francis Hutcheson and Cesare Beccaria. Especially Hutcheson [\ref{hutcheson28}] had already defined $\textit{good}$ as pleasure and $\textit{good objects}$ as objects that create pleasure. The novelty of Bentham was to treat pleasure as a measurable quantity, thus making the utilitarian doctrine directly applicable to issues like tax policies and legislation. Indeed, not only did Bentham argue that individual pleasure was measurable, but also that happiness of different people could be compared. Stark [\ref{stark52}] cited in his article Bentham's writings in the following way:

\begin{quote} \footnotesize
Fortunes unequal: by a particle of wealth, if added to him who has least, more happiness will be produced, than if added to the wealth of him who has most (vol. 1, p. 103).
\end{quote}

\noindent
Stark [\ref{stark52}] continues:

\begin{quote} \footnotesize
The quantity of happiness produced by a particle of wealth (each particle being the same magnitude) will be less and less every particle (vol. 1, p. 113).
\end{quote}

It is easy to see how this last concept and the well-known idea of decreasing marginal utility are related.

In the pioneering work of Jevons [\ref{jevons71}], utility functions were the primitive mathematical notion to formalize and quantify Bentham's calculus of pleasure. Utility functions were tools to measure and scale the amount of well-being of human beings. It seems clear, at this point, how the starting role of utility functions was $\textit{cardinal},$\footnote{Note that before the work of Hicks and Allen [\ref{hicksallen34}], economists spoke about $\textit{measurable}$ utility and not of $\textit{cardinal}$ utility.} in the sense that utility or, better, pleasure differences were $\textit{well-founded}$ and $\textit{realistic}$ notions with a strong moral philosophy justification. 

Summing up, in the beginning, utility functions were designed for the mere purpose of a calculus of pleasure and sorrow. However, even if the philosophical concept made sense, the difficulties in the quantification of any experimental measurement of pleasure led cardinal utility theory to be seen more just like a thought process rather than a science. 

However, utility theory did not rise from philosophy alone, but it was object of study of other sciences such as statistics, with the so-called St. Petersburg paradox, and psychophysics, the study of physical stimuli and their relation to sensory reactions. These two phenomena can be considered the starting point of the law of decreasing marginal utility. It was Nicolas Bernoulli that, originally, invented what is now called the St. Petersburg puzzle, which offered the theoretical explanation for the law of decreasing marginal utility of wealth. The standard version of the puzzle is the following: a fair coin is tossed until it lands ``head" on the ground. At that point, the player wins $2^{n}$ dollars, where $n$ is the number of times the coin was flipped. How much should one be willing to pay for playing this game? In other words, what is the expected value of the game, given the probability of ``head" being 0.5? The mathematical answer is 

$$\displaystyle \sum_{i=1}^{\infty} \frac{1}{2^{i}}\cdot 2^{i}= 1 + 1 + \dots = \infty.$$ \\
The only rationale for this conundrum is that, if it makes sense to maximize expected utility and if people are willing to participate to the St. Petersburg game for only a finite amount of money, then their marginal utility as a function of wealth must be, somewhere, decreasing.

Neither Bentham nor Bernoulli thought as decreasing marginal utility as a phenomenon in need of scientific justifications. Nevertheless, this came as an immediate consequence from the psychophysical theories discovered by Weber [\ref{weber46}] and generalized by Fechner [\ref{fechner60}]. One of the most important questions posed by psychophysics is what is the functional link between different degrees of a given stimulus and a given sensation. What Weber did was an in-depth study to try to measure the $\textit{smallest detectable change}$ (also called ``just noticeable difference" or ``minimum perceptible threshold") in stimuli like heat, weight and pitch. Moreover, Fechner took this ``just noticeable difference" as a unit of measurement, constructing a scale for subjective sensations. From their studies we now have the so-called Weber's law and Fechner's laws: the former states that the relative increase of a stimulus needed to produce a just noticeable change is constant and, the latter, that the magnitude of sensation is a logarithmic function of the stimulus. 

In conclusion, if wealth is a stimulus, then Benthamian utility must be the corresponding sensation. In this light, St. Petersburg puzzle can be seen as just one materialization of these laws.

At the end of 19th century, the marginalist revolution paved the way for an ordinal approach to the notion of utility. In fact, this was because one of the main economic problems of late 19th century was the need of a theory of demand. One of the leading figures that founded neoclassical theory with scientific and analytic rigor was Vilfredo Pareto. Pareto is considered the father of the so-called $\textit{ordinal approach}$. It was a notion of utility that was purely comparative and it left out from the theory the initial idea for which utility theory was developed: the existence of psychophysical and physiological substrates. Pareto's theory was so successful that was considered a revolution in the notion of utility. The ordinal approach was extremely successful because it solved the classic consumer problem based on indifference curves, and the notion of utility had a central role in its construction. The key aspect was the replacement of marginal utility - a notion that was meaningless in an ordinal approach - with the trick of marginal rate of substitutions along indifference curves. 

Interesting are the writings of Francis Edgeworth [\ref{edgeworth81}] and Pareto [\ref{pareto06}], starting from very different assumptions and arriving at different conclusions. Edgeworth's main contribution can be summarized in the synthesis of Bentham's utilitarianism and Fechner's psychophysics: his ideas were based on the unit of utility seen as a just perceivable increment of pleasure. Moreover, he was interested also in an $\textit{inter}$-personal unit of utility to be able to carry out welfare comparisons among people. Edgeworth was completely aware of the impossibility of testing these implications, but he was a strong supporter of the idea of possible comparisons of happiness among people.

Pareto, on the other hand, denied Edgeworth's intuition of comparisons of utility. Instead, Pareto [\ref{pareto06}] reckoned the theoretical possibility of a cardinal notion of utility, seen as the limit of the purely comparative notion he developed. Nevertheless, he also argued that such a notion of perfect precision is not attainable and that pleasure is only imperfectly measurable.\footnote{However, Pareto [\ref{pareto55}] writes: ``There is no reason for not accepting it [cardinal utility], with the reservation that it must be verified by the results deduced from it (p. 73)."} Summing up, Edgeworth's and Pareto's ways of conceiving measurable utility must be differentiated and utility theory is still today based on the Paretian notion mainly because of its use in the theory of demand and in the general equilibrium theory.

In 1950, ordinalism was the well-established mainstream ideology in utility theory and the cardinal notion of utility was almost completely abandoned. Nevertheless, the purely comparative approach was not convincing everyone, mainly because people's introspection suggested the existence something more. One of the main supporters of cardinalism was Maurice Allais who explicitly wrote in [\ref{allais94}]:

\begin{quote}\footnotesize
The concept of cardinal utility [...] has almost been rejected the literature for half a century. This rejection, based on totally unjustified prejudices, deprived economic analysis of an indispensable tool (p. 1).
\end{quote}

Allais [\ref{allais94}] admits that the theory of general economic equilibrium can be fully described in an ordinal world, but he immediately lists a series of theories that cannot be adequately developed without a rigorous and well-defined concept of cardinal utility and interpersonal comparisons. Some examples are the theory of dynamic evolution of the economy, the theory of fiscal policy, of income transfers, of collective preferences, social welfare analysis and political choices, of risk, of insurance and the theory of cooperative games. Then, Allais goes even further in his defense of cardinal utility, arguing that even the theory of demand could become more intuitive - and with a simpler exposition - if we could appeal to a notion of $\textit{intensity of preferences}$. In any case, as long as the conclusions of price theory do not change significantly using the ordinal and the cardinal approach, we should prefer the purely comparative approach by Occam's razor. But the problems with group decision making, social choice theory and cooperative game theory still cannot be solved. Indeed, while classical economists considered distributional problems as a fundamental part of economic science, the ordinalist approach to utility theory refused completely to deal with questions that involved interpersonal comparisons of welfare. Economists became more interested in positive statements, rather than normative ones and the accent was put on efficiency, rather than equity. This was the case of the optimum allocations in the sense of Pareto. For a complete overview of the main issues of welfare economics, the main problems with an ordinal approach and the main literature, we refer the interested reader to Sen [\ref{sen82}].

Hence, one of the main problems, still in the 21st century, is how is it possible to understand the intuitive tool of introspection to develop a rigorous theory that economists can apply in their models and explain economic phenomena. The solution does not exists yet. In the last years, the issue started getting the attention of few decision theorists mainly because of the powerful developments in the field of neuroscience and the new discipline of neuroeconomics. These fields of cognitive sciences are going into the direction of overcoming the main difficulty of the founders of utilitarianism: the difficulty of carrying out experiments on pleasure and pain and the construction of a rigorous and well-defined scale of pleasure. Nevertheless, nothing is clear yet, mainly because also the theoretical concept of cardinal utility is still vague. Cardinal utility is still used as a name for a large number of formally distinct concepts and it misses a precise and well-established definition that can be applied in decision-theoretic models.

During the 20th century a lot of methodologies to try to define a concept of $\textit{measurement}$ of human sensation have been defined.

\begin{mydef}
A $\emph{scale}$ is a rule for the assignment of numbers to aspects\footnote{For example: hardness, length, volume, density, $\dots$} of objects or events.
\end{mydef}

The result was the development of a full taxonomy of scales, with scales that differ in terms of higher precision of measurement. For an extensive treatment of the theory of measurement we refer the interested reader to Krantz et alii [\ref{measurement71}].

The issue of having a rigorous definition for cardinal utility was not solved by the theory of measurement. It was just translated in a different language: what is the suitable scale for measuring a given aspect? The definition of a unit of measurement for utility was not an easy task to solve. Even in physics, where experiments can be carried out with relatively high precision, the way a unit of measurement is defined is not perfect. One meter was originated as the $1/10$-millionth  of the distance from the equator to the north pole along a meridian through Paris. Then, the International Bureau of Weights and Measures, founded in 1875, defined the meter as the distance of a particular bar made by platinum and iridium kept in Sèvres, near Paris. More recently, in 1983, the Geneva Conference on Weights and Measures defined the meter as the distance light travels, in a vacuum, in 1/299,792,458 seconds with time measured by a cesium-133 atomic clock which emits pulses of radiation at very rapid and regular intervals.

Increases in science allow the unit of measurement to be duplicated with a better and better level of precision. The comparison with the unit of measurement of the quantity $\textit{utility}$ can be carried out with the philosophical question whether it is, for some esoteric reason, intrinsically impossible to measure human beings' pleasure or whether economic science and neuroeconomics are so underdeveloped that we still have very poor precision in measuring human felicity.\footnote{Some authors, like Ellingsen [\ref{ellingsen94}], are certain, instead, that the philosophical question of whether utility is intrinsically measurable or not is a spurious one, mainly because they see the issue of ``measurement" as a concept that is always invented and never discovered. In this light, our question can be rephrased as whether it is possible to define a correct notion of measurement that allows some kind of intrapersonal and interpersonal utility comparisons.}

The same comparison can be done with light (or heat, color and wave lengths, as it is mentioned by von Neumann and Morgenstern in [\ref{vnm47}]). For example, temperature was, in the original concept, an ordinal quantity as long as the concept $\textit{warmer}$ was known. Then, the first transition can be identified with the development of a more precise science of measurement: thermometry. With thermometry, a scale of temperature that was unique up to linear transformations was constructed. The main feature was the association of different temperatures with different classes of systems in thermal equilibrium. Classes like these were called $\textit{fixed points}$ for the scale of temperatures. Then, the second transition can be associated with the development of thermodynamics, where the absolute zero was fixed, defining a reference point for the whole scale. In physics, these phenomena had to be measured and the individual had to be able to replicate results of such measurements every time. The same may apply to decision theory and the notion of utility, someday. At the moment, the issue remains unclear, even if even Pareto was not completely skeptical about the first transition from an ordinal purely comparative approach to that of an equality relation for utility differences. Von Neumann and Morgenstern point out in [\ref{vnm47}] that the previous concept is based on the same idea used by Euclid to describe the position on a line: the ordinal utility concept of preference corresponds to Euclid's notion of $\textit{lying to the right of}$ and the derived concept of equality of utility differences with the geometrical congruence of intervals. 

Hence, the main question becomes whether the derived order relation on utility differences can be observed and reproduced. Nobody can, at the moment, answer this question.

\subsection{Axiomatization of utility theories}

In 1900, at the International Congress of Mathematicians in Paris, David Hilbert announced that he was firmly convinced that the foundation of mathematics was almost complete. Then, he listed 23 problems to be solved and to give full consistency to mathematics. All the rest was considered, by him, just details. Some of the problems consisted in the axiomatizations of some fields of mathematics. Indeed, at the beginning of the 20th century, the idea of being able to solve every mathematical problem led mathematicians to try develop all mathematical theory from a finite set of axioms. The main advantage of the axiomatic method was to give a clean order and to remove ambiguity to the theory as a whole. Axioms are the fundamental truths by which it is possible to start modeling a theory. The careful definition of them is critical in the development of a theory that does not contain contradictions. 

As a result, almost all fields of science started a process of axiomatization, utility theory as well. The ordinal Paretian revolution was the fertile environment where preferences started to be seen as primitive notions. Preference relations began to be formalized as mathematical order relations on a set of alternatives $X$ and became the starting point of the whole theory of choice. As a result, utility functions became the derived object from the preference relations. The mainstream notion of ordinal (Paretian) utility reached its maturity with the representation theorems by Eilenberg [\ref{eilenberg41}] and Debreu [\ref{debreu54}], [\ref{debreu64}]. Subsequent work in decision theory shifted from decision theory under certainty to choice problems under uncertainty, with the pioneering article of Ramsey [\ref{ramsey26}] on the ``logic of partial belief." In short, Ramsey [\ref{ramsey26}] stated the necessity of the development of a purely psychological method of measuring both probability and beliefs, in strong contradiction with Keynes' probability theory. Some years after, the milestone works of von Neumann and Morgenstern [\ref{vnm47}] and Savage [\ref{savage54}] gave full authority to decision theory under uncertainty.

One of the first treatments of preference relations as a primitive notion can be identified with Frisch [\ref{frisch26}], in his 1926 paper. Ragnar Frisch was also the first to formulate an axiomatic notion of utility difference. Hence, two kinds of axioms were postulated by him: the first ones - called ``axioms of the first kind" - regarded the relation able to rank alternatives in a purely comparative way, while the second axioms - named ``axioms of the second kind" - reflected a notion of $\textit{intensity of preference}$ and allowed utility differences to be compared. So, in parallel to the axiomatization of ordinal utility, also cardinal utility axiomatizations started to grow.

Frisch's article did not have the deserved impact in the academic arena, mainly because his article was written in French and published in a Norwegian mathematical journal. Hence, the full mathematical formalization of these two notions of preference axioms resulted almost ten years later from the 1930s debate by Oskar Lange [\ref{lange34}] and Franz Alt [\ref{alt36}]. Lange [\ref{lange34}] defined an order relation $\succ$ on the set of alternatives $X$ with the meaning that, for any two alternatives $x,y \in X$, $x \succ y$ reads ``$x$ is strictly preferred to $y$." Then, a corresponding relation $\mathcal{P}$ on ordering differences is assumed with the meaning that, for any $x,y,z,w \in X$, $xy \mathcal{P} zw$ reads ``a change from $y$ to $x$ is strictly preferred than a change from $w$ to $z$.

\noindent
More formally: 

\begin{equation}\label{langeeq1} 
x \succ y  \iff u(x) > u(y) \text{ for all } x,y \in X \hspace{3.25cm} \end{equation} 
\begin{equation}\label{langeeq2} 
xy \mathcal{P} zw  \iff u(x) - u(y) > u(z) - u(w) \text{ for all } x,y,z,w \in X \end{equation}
\\
\noindent
The main theorem of Lange [\ref{lange34}] can be stated as follows:

\begin{theorem}
If there exists a differentiable utility function $u: \R \rightarrow \R$ such that (\ref{langeeq1}) and (\ref{langeeq2}) hold, then only positive affine transformations of that utility function represent the given preferences $\succ$ and $\mathcal{P}$.
\end{theorem}

It is immediate to see that Lange [\ref{lange34}] provides only necessary conditions for a utility function representation of preference relation. Moreover, it is relatively easy to see that the assumption of differentiability of $u$ can be largely relaxed. Hence, the issue becomes whether it is possible to find sufficient conditions on the preference relations under which Lange's utility function - a cardinal utility function - exists. This was done by Franz Alt in his 1936 article [\ref{alt36}]. Alt postulated seven axioms that guaranteed sufficient and necessary conditions for the existence of a continuous utility function - unique up to positive affine transformations - based on a preference relation and a utility-difference ordering relation. In his set of axioms, Alt defined a notion that can be understood as the set of alternatives $X$ to be connected. 

With Frisch's pioneering work of 1926 and 1930s debate by Lange and Alt, the modern ingredients of cardinal utility axiomatization such as equations (\ref{langeeq1}) and (\ref{langeeq2}) and connectedness of the domain of alternatives $X$ started to be formalized. In those years, a lot of different axiomatic models were studied, till the article of the famous philosopher of science Patrick Suppes and his doctoral student Muriel Winet [\ref{suppeswinet55}]. In their 1955 paper, Suppes and Winet developed an abstract algebraic structure of axioms for cardinal utility, called a $\textit{difference structure}$, in line with old Frisch's ideas and Lange's formalization: not only are individuals able to ordinally rank different alternatives, but they are also able to compare and rank utility differences of alternatives. Indeed, Suppes and Winet cited the work of Oskar Lange on the notion of utility differences and stood in favor of the intuitive notion of introspection, elevating it to not just a mere intuition, but as a solid base where to build a notion of utility differences. Suppes and Winet continued their article saying that, up to 1950s, no adequate axiomatization for intensity comparison had been given. Hence, as Moscati [\ref{moscati19}] nicely highlights, they were probably unaware of Alt's representation theorem and this was probably due to the fact that Alt [\ref{alt36}] was published in German in a German journal. Suppes and Winet postulated 11 axioms in total, some on the set of alternatives $X$ and others on the two order relations,\footnote{The conditions these axioms impose are analogous to the conditions defined by Alt [\ref{alt36}]: completeness, transitivity, continuity, and some form of additivity for the two order relations, and an Archimedean property on the quaternary relation.} providing sufficient and necessary condition for a cardinal utility representation, unique up to positive affine transformations. Another approach to the field of axiomatization of cardinal utility was taken twenty years later by Lloyd Shapley. While axiomatizations à la Suppes and Winet started developing a set-theoretic abstract structure, Shapley substituted the usual long list of postulates with strong topological conditions both on the domain of alternatives $X$ and on the topology induced by the order relations. Shapley [\ref{shapley75}] constructed a cardinal utility function $u$ satisfying some consistency axioms between the orders and assuming the domain of $u$ to be a convex subset of the real line. We will enter into the details of Shapley [\ref{shapley75}] in the next chapters.

In conclusion, the notion of cardinal utility has always suffered a lack of conceptual precision in its whole history and, for some authors like Ellingsen [\ref{ellingsen94}], it can be even considered the main reason why scientists have disagreed over whether pleasure can be measured or not.\footnote{Ellingsen [\ref{ellingsen94}] writes about a ``fallacy of identity" and ``fallacy of unrelatedness."} What is certain is that the history of cardinal utility, a part from some sporadic articles, has been a persistent failure, mostly in its applications to economic theory. While the main reason can be probably identified with the almost total absence of any rigorous and proven experimental measurement of pleasure, it is fair to observe that part of its failure must be given to the strong reluctant opinion of the mainstream ordinal ``party." In fact, a large class of economists classify as ``meaningless" even the mere introspective idea of a comparison of utility difference, and not just the concept itself, when formalized in a purely comparative environment. This position is shown to be, with a gentle expression, ``epistemological laziness." We should always remember that no real progress in economic science can be derived from purely abstract reasoning, but only from the combined effort of empirical measurements with theoretical analysis, always under the wise guide of the compass of history and philosophy.

\section{Preliminary results}

The aim of this chapter is twofold. On one side, we introduce the mathematical framework that enable us to represent intensity comparisons that a decision maker feels about the desirability of different alternatives. For this aim, we follow the construction of Suppes and Winet [\ref{suppeswinet55}] and Shapley [\ref{shapley75}]. On the other side, we state and prove a list of lemmas that will be used in Theorem \ref{thm1} and that allow us to generalize Shapley's proof to a connected and separable subset of a topological space.

\subsection{Basic definitions}

\begin{mydef}
A $\emph{relation}$ on a set $X$ is a subset $\succsim$ of the cartesian product $X \times X$, where $x \succsim y$ means $(x,y) \in \hspace{0.15cm} \succsim$.
\end{mydef}

In decision theory, $\succsim$ is usually called a $\emph{preference}$ relation, with the interpretation that, for any two elements $x, y \in X$, we write $x \succsim y$ if a decision maker either strictly prefers $x$ to $y$ or is indifferent between the two.

\begin{mydef}
An $\emph{equivalence relation}$ on a set $X$ is a relation $\mathcal{R}$ on $X$ that satisfies \begin{enumerate}
\item[1)] Reflexivity: for all $x \in X$, we have $x \mathcal{R} x$.
\item[2)] Symmetry: for any two elements $x,y \in X$, if $x \mathcal{R} y$, then $y \mathcal{R} x$.
\item[3)] Transitivity: for any three elements $x,y$ and $z \in X$, if $x \mathcal{R} y$ and $y \mathcal{R} z$, then $x \mathcal{R} z$.
\end{enumerate}
\end{mydef}

\begin{mydef}
A relation $\succsim$ on a set $X$ is called a $\emph{total order relation}$ (or a $\emph{simple order}$, or a $\emph{linear order}$) if it has the following properties: \begin{enumerate}
\item[1)] Completeness: for any two elements $x,y \in X$, either $x \succsim y$ or $y \succsim x$ or both.
\item[2)] Antisymmetry: for any two elements $x,y \in X$, if $x \succsim y$ and $y \succsim x$, then $x=y$.
\item[3)] Transitivity: for any three elements $x,y$ and $z \in X$, if $x \succsim y$ and $y \succsim z$, then $x \succsim z$.
\end{enumerate}
\end{mydef}

Note that if $\succsim$ is complete, then it is also reflexive. The relation $\succsim$ induces, in turns, two other relations. Specifically, for any two elements $x,y \in X$ we write:
\begin{enumerate}
\item[(i)] $x \succ y$ if $x \succsim y$ but not $y \succsim x$.
\item[(ii)]$x \sim y$ if $x \succsim y$ and $y \succsim x$.
\end{enumerate}

It is easy to see, indeed, that if $\succsim$ is reflexive and transitive, then $\sim$ is an equivalence relation. Given an equivalence relation $\sim$ on a set $X$ and an element $x \in X$, we define a subset $E$ of $X$, called the $\emph{equivalence class}$ determined by $x$, by the equation $$E:= \{y \in X : y \sim x \}$$ Note that the equivalence class $E$ determined by $x$ contains $x$, since $x\sim x$, hence $E$ is usually denoted as $[x]$. We will denote $\equivclass{X}$ the collection $\{[x] : x \in X\}$ of all equivalence classes, which is a partition of $X$: each $x \in X$ belongs to one, and only one, equivalence class. In decision theory, an equivalence class is often called an $\emph{indifference curve}$.

\begin{mydef}
A relation $\succsim$ on a set $X$ is called a $\emph{weak order}$ if it is complete and transitive.
\end{mydef}

The problem of finding a numerical representation for a preference relation $\succsim$, $\ie$ an order isomorphism between a generic set $X$ and $\R$, has been widely studied by mathematicians and is a familiar and well-understood concept. Such an order isomorphism is called, in decision theory, a $\emph{utility function}$. More formally: 

\begin{mydef}
A real-valued function $u:X \rightarrow \R$ is a $\emph{(Paretian) utility function}$ for $\succsim$ if for all $x,y \in X$ we have $$x \succsim y \iff u(x) \ge u(y)$$
\end{mydef}

Utility functions ``shift" the pairwise comparisons that characterize the order relation $\succsim$ and its properties in the more analytically convenient space of the real numbers. Nevertheless, as a result, the only thing that is preserved is the order, and the real numbers that are images of the utility function cannot be interpreted as a scale where the decision maker can compare different intensities about the single desirability of any two alternatives $x,y \in X$. What is important is the ranking given by the real numbers, according to the usual order of the ordered field $(\R,\ge)$. Indeed, one can easily prove that every strictly increasing transformation of a utility function is again a utility function. For this reason, utility functions are called $\emph{ordinal}$ and their study belong to what is called $\emph{ordinal utility theory}$. The main problem of $\emph{ordinal utility theory}$ is to study sufficient and necessary conditions under which a relation $\succsim$ admits a utility representation. The original reference can be identified with Cantor [\ref{cantor95}], but the result has been adapted by Debreu [\ref{debreu54}]. 

In addition, to be able to solve optimization problems, one of the properties that is desirable to have is $\emph{continuity}$ of the utility function. Debreu [\ref{debreu54}] is the first to state the theorem in the way we are going to. Nevertheless, he proved it making explicit reference to Eilenberg [\ref{eilenberg41}]. We state here a version of this very well-known theorem.

\begin{mydef}
A weak order $\succsim$ on a set $X$ is said to be $\emph{continuous}$ if, for every $y \in X$, the sets $\{x \in X : x \succ y \}$ and $\{x \in X : x \prec y \}$ are open.\footnote{Note that this is the usual order topology on $X$.}
\end{mydef}

\begin{theorem}[\textbf{Eilenberg}] \label{Eilenbergthm}
Let $\succsim$ be a complete and transitive relation on a connected and separable topological space $X$. The following conditions are equivalent: \begin{enumerate}
\item[(i)] $\succsim$ is continuous.
\item[(ii)] $\succsim$ admits a continuous utility function $u: X \rightarrow \R$.
\end{enumerate}
\end{theorem}

One of the biggest theoretical problems of $\emph{ordinal utility theory}$ is that the expression $$u(x)-u(y)$$ is a well-defined real number thanks to the algebraic properties of $\R$, but it is meaningless in term of the interpretation of a $\emph{difference of utility}$ of two alternatives $x,y \in X$. In other words, a Paretian utility function does not have an intrinsic introspective psychological notion of $\emph{intensity}$ of the preferences. An immediate corollary of this remark is that the concept of $\emph{marginal utility}$ (and what is known under the Gossen's law of decreasing marginal utility), based on the notion of different quotient, is meaningless. More formally, the expression 

$$\frac{du(x)}{dx} = \lim_{h \rightarrow 0} \frac{u(x +h) - u(x)}{h}$$
\vspace{0.05cm}

\noindent
has no meaning in this setting.
Nevertheless, the concept of marginal utility has been a milestone in economic theory, proving that this notion deserves an adequate theoretical foundation.

\subsection{An overview on measurable utility theory}

Let $X$ be a set of alternatives. Pairs of alternatives $(x,y) \in X \times X$ are intended to represent the prospect
of replacing alternative $y$ by alternative $x$, that can be read as ``$x$ in lieu of $y$".
Define the binary relation $\succcurlyeq$ on $X \times X$ called $\textit{intensity preference}$ with the following interpretation: for any two pairs $(x,y)$ and $(z,w)$ in $X \times X$, $$(x,y) \succcurlyeq (z,w)$$ is intended to mean that getting $x$ over $y$ gives at least as much added utility as getting $z$ over $w$ or (if $y \succsim x$) at most as much added sadness.
As a result, our decision maker is endowed with a weak order preference relation $\succsim$ on alternatives and an intensity preference relation $\succcurlyeq$ on pairs of alternatives. 

Shapley [\ref{shapley75}] proves his theorem assuming $X$ to be a convex subset of $\mathbb{R}$. As a result, the proof exploits the full algebraic power of the ordered field and the topological properties of the linear continuum. Our aim is to generalize the set of alternatives $X$ to a connected and separable subset of a topological space, ordered with the binary relations $\succsim$ and $\succcurlyeq$ and with the order topology induced by the weak order $\succsim$.

\noindent
We assume the following axioms for $\succsim$ and $\succcurlyeq$, as in Shapley [\ref{shapley75}]. 
\noindent
\begin{axiom} \normalfont{For all} $x,y,z \in X$ \normalfont{we have} $(x,z) \succcurlyeq (y,z) \iff x \succsim y$. \end{axiom}

Axiom 1 (henceforth $\bold{A1}$) is an assumption of $\textit{consistency}$ between the two orderings because it implies that the decision maker prefers to exchange $z$ with $x$ instead of $z$ with $y$ if and only if she prefers $x$ to $y$. Together with $\bold{A1}$ we can formulate a dual version of $\textit{consistency}$, $\bold{A1'}$, that can be derived from the whole set of axioms we are going to assume later.\footnote{We mention $\bold{A1'}$ as a form of axiom only because in this way we can refer to it in the proof of Theorem \ref{thm1}, but we never assume it formally. A proof of it will be formulated forward with Lemma \ref{lemmaaprime}.}

\begin{customthm}{1{$'$}}\label{onestar}
 \normalfont{For all} $x,y,z \in X$ \normalfont{we have} $(z,x) \succcurlyeq (z,y) \iff y \succsim x$.
\end{customthm}

We now introduce the main object of this thesis: a joint real-valued representation for the two orders $\succsim$ and $\succcurlyeq$.

\begin{mydef}
A real-valued function $ u:X \rightarrow \mathbb{R}$ is a \normalfont{measurable utility function} for $(\succsim, \succcurlyeq)$ if for each pair $x,y \in X$ 
\begin{equation}
x  \succsim y \iff u(x) \ge u(y)
\end{equation}
and if, for each quadruple $x,y,z,w \in X$
\begin{equation}
(x,y) \succcurlyeq (z,w) \iff u(x) - u(y) \ge u(z) - u(w).
\end{equation}
\end{mydef}

The $\textit{measurable}$ terminology has nothing to do with measure theory, but it refers to what is known as $\textit{measurement theory}, \textit{i.e.}$ the field of science that established the formal foundation of quantitative measurement and the assignment of numbers to objects in their structural correspondence. Indeed, not only is a measurable utility function able to rank pairs of alternatives according to a preference relation, but it also represents the idea of magnitude and intensity of the preference relation among alternatives. Therefore, the numerical value $u(x)$ that a measurable utility function assigns to the alternative $x$ is assuming the role of a particular $\textit{unit of measurement}$ for pleasure, that we call $\textit{util}$.

Recall that an ordinal utility function $u$ is unique up to strictly monotone transformations $f:\text{Im(u)} \rightarrow \mathbb{R}$. Hence, a measurable utility function is not ordinal. Nevertheless, it is unique up to positive affine transformations. Recall that a positive affine transformation is a special case of a strictly monotone transformation of the following form $f(x)=\alpha x + \beta$, with $\alpha >0$ and $\beta \in \mathbb{R}$. Positive affine transformations are order-preserving thanks to $\alpha >0$. 

\begin{proposition}
A measurable utility function $ u:X \rightarrow \mathbb{R}$ for $(\succsim, \succcurlyeq)$ is unique up to positive affine transformations.
\end{proposition}

\begin{proof}
If $\overline{u}(x)= \alpha u(x) + \beta$ then we have $$x  \succsim y \iff u(x) \ge u(y) \iff \overline{u}(x)= \alpha u(x) + \beta \ge \alpha u(y) + \beta=\overline{u}(y)$$
and
\begin{align*}
(x,y) \succcurlyeq (z,w) & \iff u(x) - u(y) \ge u(z) - u(w)\\
 & \iff \overline{u}(x) - \overline{u}(y) = \alpha[u(x) -  u(y)] \ge \alpha[u(z) -  u(w)] = \overline{u}(z) - \overline{u}(w).
\end{align*}
\end{proof}
\noindent
As a result, $u$ and $\overline{u}$ are two utility representations for $(\succsim, \succcurlyeq)$. The whole class of utility functions that are unique up to positive affine transformations are called $\textit{cardinal}$. Measurable utility functions are, therefore, cardinal and pertain to the so-called $\textit{cardinal utility theory}.$

Other two axioms ($\mathbf{A2}, \mathbf{A3}$) we need to introduce are the following:

\begin{axiom} \normalfont{For all} $x,y,z,w \in X$ \normalfont{we have} $(x,y) \sim (z,w) \iff (x,z) \sim (y,w)$. \end{axiom}
\begin{axiom} \normalfont{For all} $x,y,z,w \in X$ \normalfont{the set} $$\{(x,y,z,w) \in X \times X \times X \times X : (x,y) \succcurlyeq (z,w) \}$$ is closed in the product topology. \end{axiom}

Axiom 2 is a ``crossover" property that characterizes difference comparisons of utility, while Axiom 3 is a technical assumption defining the order relation $\succcurlyeq$ as $\textit{continuous}.$

Shapley [\ref{shapley75}] proves his theorem on a domain of alternative outcomes that is a nonempty, convex subset $\mathcal{D}$ of the real line where the preference order coincides with the total order of $(\mathbb{R}, \ge)$. Moreover, $\succcurlyeq$ is assumed to be a weak order on $\mathcal{D} \times \mathcal{D}$ such that $\mathbf{A1}$, $\mathbf{A2}$ and $\mathbf{A3}$ are satisfied.

\begin{theorem}[\textbf{Shapley}] \label{Shapleythm}
There exist a utility function $u: \mathcal{D} \subseteq \mathbb{R} \rightarrow \mathbb{R}$ such that
\begin{equation}
x  \ge y \iff u(x) \ge u(y)
\end{equation}
and
\begin{equation}
(x,y) \succcurlyeq (z,w) \iff u(x) - u(y) \ge u(z) - u(w)
\end{equation}
for all $x,y,z,w \in \mathcal{D}.$ Moreover, this function is unique up to a positive affine transformation.
\end{theorem}
\noindent
The theorem is stated as a sufficient condition, which is the most difficult part to prove. The necessary condition of the theorem is easily proved and we state it here as a proposition.

\begin{proposition}
If the pair $(\ge, \succcurlyeq)$ has a continuous measurable utility function $u: \mathcal{D} \subseteq \mathbb{R} \rightarrow \mathbb{R}$, then $\ge$ is complete and transitive, $\succcurlyeq$ is complete, transitive, continuous $(\mathbf{A3})$ and satisfies the crossover axiom $(\mathbf{A2})$, and jointly $\ge$ and $\succcurlyeq$ satisfy the consistency axiom $(\mathbf{A1})$.
\end{proposition}

Shapley's construction of the measurable utility function of Theorem \ref{Shapleythm} is extremely elegant, but has the drawback of being too specific as $u$ is defined on a convex subset of $\mathbb{R}$. On the other side of the spectrum, as mentioned in the first chapter, the field of utility axiomatization has been prolific in the 20th century and a copious number of cardinal-utility derivations from preference-intensity axiomatizations were published. One of the most important papers on this issue was the one published in 1955 by Patrick Suppes and Muriel Winet. Recalling what described before, Suppes and Winet [\ref{suppeswinet55}] advanced an axiomatization of cardinal utility based on the assumption that individuals are not only able to rank the utility of different alternatives, as is assumed in the ordinal approach to utility, but are also capable of ranking the differences between the utilities of commodities. Nevertheless, their 11 axioms on an abstract algebraic structure were not fully satisfactory in terms of generality: it was too general. Indeed, some of their axioms can be derived in Shapley [\ref{shapley75}], thanks to the topological properties of $\mathbb{R}$.

The aim of this research is to settle somewhere in between, finding a representation theorem for cardinal utility function (in particular, a measurable one) keeping the elegance of Shapley's proof and generalizing the domain of alternatives into the direction of Suppes and Winet [\ref{suppeswinet55}]. We will state and prove a representation theorem for a measurable utility function $u:X \rightarrow \mathbb{R}$ where $X$ is a connected and separable subset of a topological space, $\succsim$ and $\succcurlyeq$ are weak orders and they satisfy $(\mathbf{A1})$, $(\mathbf{A2})$ and $(\mathbf{A3})$. Before doing this, we need to state and prove some topological preliminary results that will be used in Theorem \ref{thm1}.\footnote{We thank Dr. Hendrik S. Brandsma for providing a feedback and insightful comments.}

\subsection{A few basic lemmas}

\begin{mydef}
Let $X$ be a topological space. $X$ is $\emph{connected}$ if it cannot be separated into the union of two disjoint nonempty open subsets. Otherwise, such a pair of open sets is called a $\emph{separation}$ of $X$.
\end{mydef}

\begin{mydef}
Let $X$ be a topological space. $X$ is $\emph{separable}$ if there exists a countable dense subset. A $\emph{dense}$ subset $D$ of a space $X$ is a subset such that its closure equals the whole space, $\textit{i.e.}$ $\overline{D} = X$.
\end{mydef}

\begin{mydef}
A totally ordered set $(L, \succsim)$ having more than one element is called a $\emph{linear}$ $\emph{continuum}$ if the following hold:
\begin{enumerate}
\item $L$ has the least upper bound property.
\item If $x \succ y$, there exists $z$ such that $x \succ z \succ y$
\end{enumerate}
\end{mydef}

We recall that a $\textit{ray}$ is a set of the following type
$(-\infty, a)=\{x \in L: x \prec a\}$ and
$(-\infty, a]=\{x \in L: x \precsim a\}$ in the case $L$ does not have a minimum. In the case $L$ does have a minimum we write $\left[x_{m}, a\right)=\left\{x \in L: x_{m} \precsim x\prec a\right\}$ and $\left[x_{m}, a\right]=\left\{x \in L: x_{m} \precsim x \precsim a\right\}$. Analogously for the sets $(a,+\infty),[a,+\infty),\left(a, x_{M}\right],\left[a, x_{M}\right]$, where $x_{M}$ is the maximum of $L$ in the case it existed.\footnote{Note that in decision theory, $\textit{rays}$ of a set $X$ equipped with a reflexive and transitive binary relation $\succsim$ are usually denoted with the following notation $L(a, \succsim):=(-\infty, a]=\{x \in X: x \precsim a\}$ and $U(a, \succsim):= [a,+\infty)= \{x \in X: x \succsim a\}$, $L(a, \succ):=(-\infty, a)$ and $U(a, \succ):=(a,+\infty).$}

Given $A \subseteq X$, an element $y \in X$ is an $\emph{upper bound}$ for a set $A$ if $y \succsim x$ for all $x \in A$. It is a $\emph{least upper bound}$ for $A$ if, in addition, it is the minimum of the set of all upper bounds of $A$, that is if $y' \succsim x$ for all $x \in A$ then $y' \succsim y$. If $\succsim$ is antisymmetric, the least upper bound is unique and is denoted $\sup{A}$. The $\emph{greatest lower bound}$ is defined analogously and denoted $\inf{A}$.

\begin{lemma}\label{linearcontinuum}
Let $\succsim$ be a total order on a connected set $X$. Then, $X$ is a linear continuum in the order topology.\footnote{Note that the converse holds as well: $\succsim$ is a total order on a connected set $X$ if and only if $X$ is a linear continuum in the order topology.} 
\end{lemma}

\begin{proof}
Suppose that $a$ and $b$ are two arbitrary but fixed elements of $X$ such that $a \prec b.$ If there is no element $c \in X$ such that $a \prec c \prec b$, then $X$ is the union of the open rays $(-\infty, b)=\{x \in X: x \prec b\}$ and $(a,+\infty)=\{x \in X: a \prec x\}$ both of which are open sets in the order topology and are also nonempty, as the first contains $a$, while the second contains $b$. But this contradicts the fact that $X$ is connected, so there must exists an element $c \in X$ such that $a \prec c \prec b.$

Now, to show the least upper bound property, let $A$ be a nonempty subset of $X$ such that $A$ is bounded above in $X$. Let $B$ be the set of all the upper bounds in $X$ of set $A$, $\textit{i.e.}$ $$B:=\{b \in X: b \succsim a \text { for every } a \in A\}$$ which is nonempty. All we need to show is that $B$ has the least element. If $B$ has a smallest element (or $A$ has a largest element, which would then be the smallest element of $B$), then that element is the least upper bound of $A$. 

Let us assume, instead, that $B$ has no smallest element. Then, for any element $b \in B$, there exists an element $b^{\prime} \in B$ such that $b^{\prime} \prec b$, and so $b \in\left(b^{\prime},+\infty\right) \subseteq B$ with $\left(b^{\prime},+\infty\right)$ being an open set in $X$. This shows that $B$ is a nonempty open subset of $X$. Therefore, $B$ can be closed only in the case when $B=X$. But we know that $B \subset X$, since $A \subseteq X \backslash B$ and $A \neq \emptyset$, so it cannot be the case that $B=X$. Therefore, $B$ has a limit point $b_0$ that does not belong to $B$. Then $b_0$ is not an upper bound of set $A$, which implies the existence of an element $a \in A$ such that $b_0 \prec a$, we can also conclude that $b_0 \in(-\infty, a) \subseteq X \backslash B$, with $(-\infty, a)$ being an open set. This contradicts our choice of $b_0$ as a limit point of set $B.$ Therefore, the set $B$ of all the upper bounds in $X$ of set $A$ must have a smallest element, and that element is the least upper bound of $A$.
\end{proof}

Given $A \subseteq X$, we denote $\overline{A}$ or Cl$A$ the topological $\textit{closure}$ of $A$, that is defined as the intersection of all closed sets containing $A$.



From now on denote $X$ as a subset of a topological space $(\mathbf{X},\tau)$, unless otherwise stated.

\begin{lemma} \label{wlc}
Let $\succsim$ be a complete, transitive and continuous order on a connected set $X$. Given any $x,y \in X$, with $x \succ y$, we have $$x \succsim z \succsim y \Rightarrow z \in X \hspace{0.5cm} \text{for all  } z \in (\mathbf{X},\tau)$$
\end{lemma}

\begin{proof}
Suppose by contradiction that there exists $z \in \mathbf{X} \backslash X$ such that $x \succ z \succ y$. By the continuity of $\succsim$, we can partition $X$ into two nonempty disjoint open sets $\{x \in X: x \prec z\}$ and $\{x \in X: x \succ z\}$, which contradicts the connectedness of $X$.
\end{proof}

\begin{lemma}
Suppose that jointly $\succsim$ and $\wpp$ satisfy $\Aone$. If $\wpp$ is continuous , then $\succsim$ is continuous.
\end{lemma}

\begin{proof}
For all arbitrary but fixed $y, z \in X$, by $\Aone$ we have $\{x \in X : (x,z) \wpp (y,z)\} = \{x : x \succsim y\}$. By $\Athree$, the set $\{x \in X : (x,z) \wpp (y,z)\}$ is closed. Analogous is the case for $\{x : y \succsim x\}$, derived from $\Aoneprime$.
\end{proof}

\begin{lemma} \label{t1}
Fix $y \in X$, the set $\mathcal{I}_{y}:=\{x \in X : x \sim y\}$ is a closed set in $X$.
\end{lemma}

\begin{proof}
$\succsim$ is continuous, so for every $y \in X$ we have that $\{x \in X: x \succsim y \}$ and $\{x \in X: y \succsim x \}$ are closed. Pick a point $x$ such that $x \succsim y $ and $y \succsim x$, that is $x \sim y$. So we have $\{x \in X: x \sim y \} = \{x \in X: x \succsim y \} \cap \{x \in X: y \succsim x \}$ and the intersection of two closed sets is closed.
\end{proof}

Note that when $\succsim$ is antisymmetric, the set $\mathcal{I}_{y}$ is a singleton and Lemma \ref{t1} reduces to prove that $X$ satisfies the $T_1$ axiom of separation, that is every one-point set is closed. Clearly, every Hausdorff space satisfies it.

\begin{lemma} \label{sup}
Let $\succsim$ be a continuous total order on a connected set $X$. If $A \subseteq X$ is a nonempty closed set in the order topology and $A$ is bounded above (below), then $\emph{sup}A$ $\emph{(inf}A)$ belongs to $A$.\footnote{The lemma holds even in the case we relaxed connectedness. Nevertheless, we always need to assume $\sup{A}$ exists. If we do not assume the existence of the least upper bound, an easy counterexample is $\mathbb{N} \subset \mathbb{R}$ that is closed in the order topology, but $\sup{\mathbb{N}} \notin \mathbb{N}$.}
\end{lemma}

\begin{proof}
Suppose $\text{sup}A \notin A$. Then $\text{sup}A \in X\backslash A$, which is open. By definition, there exists a base element $(a,b)$ such that 
$$ \text{sup}A \in (a,b) \subseteq X\backslash A.$$
$A$ is bounded above so, by Lemma \ref{linearcontinuum}, $\sup{A}$ exists and there is an element $a^{\star}$ such that $a \prec a^{\star} \prec \sup{A}$, then $a^{\star} \in (a,b) \subseteq X\backslash A$, so $a^{\star}$ is an upper bound of $A$ smaller that $\text{sup}A$, reaching a contradiction.
In the case $X$ had a maximum, then consider the case where $\sup{A} = \max{X}$. Let $U := (x,\sup{A}]$ be a basic neighborhood of $\sup{A}$. Then, $x$ cannot be an upper bound of $A$ as $x \prec \sup{A}$. Hence, there exists an element $a \in A$ such that $x \prec a \precsim \sup{A}$. Thus, as $x$ was generic, it follows that $U \cap A \neq \emptyset$. This means that every neighborhood of $\sup{A}$ intersects $A$, that is $\sup{A} \in \overline{A}$. But $A$ is closed, hence $\sup{A} \in A$ and we can conclude $\sup{A} = \max{A}$.

The case of $\inf{A}$ is specular.
\end{proof}

Now we define the notion of convergence in any topological space. 

\begin{mydef}
In an arbitrary topological space $X$, we say that a sequence $x_1, x_2, \dots$ of points of the space $X$ $\emph{converges}$ to the point $x$ of $X$ provided that, corresponding to each neighborhood $U$ of $x$, there is a positive integer $N$ such that $x_n \in U$ for all $n \ge N$. 
\end{mydef}

Moreover, let $\succsim$ a total order. We write $x_{n} \uparrow x$ if $x_1 \precsim x_2 \precsim \dots \precsim x_{n} \precsim \dots$ and $\emph{sup}_{n}x_{n}=x$ where sup is with respect to $\precsim$. The definition $x_{n} \downarrow x$ for a $\precsim$-decreasing sequence is analogous. We say that $(x_{n})$ converges $\emph{monotonically}$ to a limit point $x$ when either $x_{n} \uparrow x$ or $x_{n} \downarrow x$.

We now prove one of the fundamental lemmas that allow us to generalize Shapley's proof to a connected and separable subset of a topological space. Note that, as long as Shapley [\ref{shapley75}] is working on $\mathbb{R}$, sequences as ``enough" to characterize the definition of convergence. This is due to the fact that there exists a countable collection of neighborhoods around every point. This is not true in general, but it is for a specific class of spaces that are said to satisfy the $\textit{first countability axiom}$.\footnote{There are far more general classes of spaces in which convergence can be fully characterized by sequences. We refer the interested reader to the notion of $\textit{Fréchet-Urysohn spaces}$ and $\textit{Sequential spaces}$.} A space $X$ is said to have a $\textit{countable basis at the point x}$ if there is a countable collection $\{U_{n}\}_{n \in \mathbb{N}}$ of neighborhoods of $x$ such that any neighborhood $U$ of $x$ contains at least one of the sets $U_n$. A space $X$ that has a countable basis at each of its points is said to satisfy the $\textit{first countability axiom}$.

In general, however, sequences are not powerful enough to capture the idea of convergence we want to capture in a generic topological space. Indeed, there could be uncountably many neighborhoods around every point, so the countability of the natural number index of sequences cannot ``reach" these points. The ideal solution to this problem is to define a more general object than a sequence, called a \textit{net}, and talk about net-convergence. One can also define a type of object called a \textit{filter} and show that filters also provide us a type of convergence which turns out to be equivalent to net-convergence. With these more powerful tools in place of sequence convergence, one can fully characterize the notion of convergence in any topological space. 

Nevertheless, we are now going to show that every connected, separable and totally ordered set $X$ satisfies the first countability axiom. In fact, we are going to prove even more. We are going to show that $X$ is metrizable, which means there exists a metric $d$ on the set $X$ that induces the topology of $X$.\footnote{A metrizable space always satisfies the first countability axiom.} We give other two definitions that will be used to prove Lemma \ref{metric}.

\begin{mydef}
Suppose $X$ is $T_1$. Then $X$ is said to be $\emph{regular}$ (or $\emph{T}_3$) if for each pair consisting of a point $x$ and a closed set $B$ disjoint from $x$, there exist disjoint open sets containing $x$ and $B$, respectively.
\end{mydef}

\begin{mydef}
If a space $X$ has a countable basis for its topology, then $X$ is said to satisfy the $\emph{second countability axiom}$, or to be $\emph{second-countable}$.
\end{mydef}

\begin{theorem}[\textbf{Urysohn metrization theorem}]
Every regular space $X$ with a countable basis is metrizable.
\end{theorem}

\begin{lemma} \label{metric}
Let $\succsim$ be a continuous total order on a connected and separable topological space $X$ in the order topology and $A \subseteq X$. We have $x \in \overline{A}$ if and only if there exists a sequence $(x_{n}) \in A^{\mathbb{N}}$ that converges monotonically to $x$.

\end{lemma}

The steps of the proof are the following:
\begin{enumerate}
\item [(\textit{i})] We show that $X$ is regular\footnote{In fact, one could prove that $X$ is also normal.} and second-countable. By the Urysohn metrization theorem, which provides sufficient (but not necessary) conditions for a space to be metrizable, there exist a metric $d$ that induces the topology of $X$.
\item [(\textit{ii})] Let $A \subseteq X$ with $X$ metrizable, then we have that $x \in \overline{A}$ if and only if there exists a sequence of points of $A$ converging to $x$.
\item [(\textit{iii})]Finally, we use the fact that in every totally ordered topological space $X$, every sequence admits a monotone subsequence. Then, if a sequence converges, all of its subsequences converge to the same limit. Thus, we can extract our monotone converging sequence.
\end{enumerate}

\begin{lemma} \label{regular}
A totally ordered topological space $X$ is regular in the order topology.
\end{lemma}

\begin{proof}
It is basic topology to prove that every totally ordered set is Hausdorff, hence it is $T_1$. Now, suppose $x\in X$ and $B$ is a closed set, disjoint from $x$. So, $x \in X\backslash B$, which is open. Then, by definition of open set, there exists a basis element $(a, b)$ such that $x \in (a, b)$ and $(a,b) \cap B = \emptyset$. Pick any $a_0 \in(a, x),$ and let $U_{1}=\left(-\infty, a_0\right), V_{1}=\left(a_0, \infty\right)$. If no such $a_0$ exists (in our case it would, by connectedness of $X$), then let $U_{1}=(-\infty, x), V_{1}=(a, \infty)$. In both cases, $U_1 \cap V_1 = \emptyset$. Similar is the case of the other side, pick $b_0 \in(x, b),$ and if that exists, denote $U_{2}=\left(b_0, \infty\right), V_{2}=\left(-\infty, b_0\right),$ and if not, let $U_{2}=(x, \infty), V_{2}=(-\infty, b)$. Again, in both cases $U_2 \cap V_2 = \emptyset$. As a result, we obtained that, in both cases, $x \in V_{1} \cap V_{2} $ with $V_{1} \cap V_{2} $ open set and $B \subseteq U_{1} \cup U_{2}$, with $U_{1} \cup U_{2}$ open set. As $V_{1} \cap V_{2}$ is disjoint from $U_{1} \cup U_{2}$, $X$ is regular.
\end{proof}

\begin{lemma}\label{secondcountable}
A totally ordered, connected and separable topological space $X$ is second-countable.
\end{lemma}

\begin{proof}
Now we find a countable basis for the order topology of $X$.
As $X$ is separable, then let $D \subseteq X$ be countable and dense in $X$, $\ie$ $\overline{D}=X$. Then, define $$\mathcal{B}:= \{(a,b) : a,b \in D \text{ with } a \prec b\}$$ together with, if there exists a minimal element $m:= \min{X}$ and a maximal element $M:= \max{X}$, the set $\{[m,a) , (a,M], a\in D \}$. In both cases, the collection $\mathcal{B}$ forms a countable base for the topology of $X$. To prove this, we show that for each open set $(a,b)$ of the order topology of $X$ and for every $x \in (a,b)$ there is an element $(a',b') \in \mathcal{B}$ such that $x\in (a',b') \subseteq (a,b)$.

Suppose $x \in (a,b) \subset X$, then the open intervals $(a,x)$ and $(x,b)$ cannot be empty by connectedness. Hence, there exist $a' \in (a,x) \cap D$ and $b' \in (x,b) \cap D$. This follows from the fact that $\overline{D} = X$ and $x \in \overline{D} = X$ if and only if every open set containing $x$ intersects $D$. Then, it follows that $x \in (a',b') \subseteq (a,b)$.  

Now, when $m$ exists, suppose $x = m$, then $x \in [m,a)$ and this set is nonempty by connectedness. Hence, there exists an element $a'' \in [m,a) \cap D$. So, it follows that $x \in [m,a'') \subseteq [m,a)$. Analogous is the case when $M$ exists.
\end{proof}

By Lemma \ref{regular} and Lemma \ref{secondcountable} , $X$ satisfies all the assumptions of the Urysohn metrization theorem, hence $X$ is metrizable (and, a fortiori, it is first-countable).

\begin{lemma} \label{sequencelemma}
Let $A \subseteq X$ with $X$ metrizable, then $x \in \overline{A}$ if and only if there exists a sequence of points of $A$ converging to $x$.
\end{lemma}

\begin{proof}
Suppose $x_n \rightarrow x$ with $x_n \in A$. Then, every neighborhood $U$ of $x$ contains a point of $A$, $\ie$ $ x \in \overline{A}$.
Conversely, we use the fact that $X$ is metrizable.\footnote{Note, once again, that here we do not need the full strength of metrizability. All we really need is a countable collection of neighborhoods around $x$. Moreover, both connectedness and separability are not necessary conditions. We refer the interested reader to the nice two-page paper of Lutzer [\ref{lutzer}], that proves a linearly ordered space $X$ is metrizable in the order topology if and only if the diagonal $\Delta :=\{(x, x): x \in X\}$ is a countable intersection of open subsets of $X \times X$, $\ie$ the diagonal is a $G_{\delta}$ set. Furthermore, this condition can be shown to be equivalent to have a $\sigma$-locally countable basis, which is a condition more in the spirit of the Nagata-Smirnov metrization theorem which requires a $\sigma$-locally finite basis.} Let $x\in \overline{A}$ and let $d$ be a metric that induces the order topology. For every $n \in \mathbb{N}$, we take the neighborhood $B_{d}(x, \frac{1}{n})$, of $x$ of radius $\frac{1}{n}$ and we choose $x_n$ to be a point such that, for all $n$, $x_n \in B_{d}(x, \frac{1}{n})\cap A$. We show $x_n \rightarrow x$. Any open set $U$ containing $x$ contains an $\epsilon$-neighborhood $B_{d}(x, \epsilon)$ centered at $x$. Choosing $N$ such that $\frac{1}{N} < \epsilon$, then $U$ contains $x_n$ for all $n \ge N$.
\end{proof}

We can finally prove Lemma \ref{metric}.
\begin{proof}
The $\textit{if}$ part comes trivially by definition. If there exists a sequence that converges (monotonically) to $x$, then $x \in \overline{A}$ by Lemma \ref{sequencelemma}.

Conversely, if $x \in \overline{A}$, then by Lemma \ref{sequencelemma} we know that there exists a sequence in $A$ converging to $x$. Now we show that, in every totally ordered set $(X, \precsim)$, every sequence from $\mathbb{N} \rightarrow (X, \precsim)$ has a monotone subsequence. Indeed, this is a property that has nothing to do with the topology of $X$.

Let $(x_{i})_{i \in \mathbb{N}}$ be a sequence with values in $X$. We say that $x_{k}$ is a $\textit{peak}$ of the sequence if $h > k \Rightarrow x_{h} \precsim x_{k}$ (we admit a slight abuse of notation here, as it would be better to call $\textit{peak}$ the index of the sequence, and not its image). We distinguish two cases: if there are infinitely many peaks, then the subsequence of peaks is an infinite non-increasing sequence and we are done. If there are only finitely many peaks, then let $i_1$ be the index such that $x_{i_{1}}$ is the successor of the last peak. Then, $x_{i_{1}}$ is not a peak. Again, we find another index $i_2 > i_1$ such that $x_{i_{2}} \succsim x_{i_{1}}$. Again, as $x_{i_{2}}$ is not a peak, we can find another index $i_3 > i_2$ such that $x_{i_{3}} \succsim x_{i_{2}} \succsim x_{i_{1}}$. Keeping defining the sequence in this way, we get, inductively, a non-decreasing sequence.

In conclusion, as by assumption we have a sequence $(x_n) \in A^{\mathbb{N}}$ converging to $x$, this sequence admits a monotone subsequence. But, if a sequence converges to a point $x$, then all of its subsequences converge to the same point $x$. Hence, there exists a sequence that converges monotonically to $x$, proving Lemma \ref{metric}.
\end{proof}

Note that Lemma \ref{metric} could have been proven just using the notion of first countability. Nevertheless, we decided to take the longer path of Urysohn metrization theorem to show how ``well-behaved" a totally ordered, connected and separable topological space can be.

\begin{lemma} \label{producttopology}
Let $(X, \succsim)$ be a topological space with the order topology. Let $\wpp$ be another order relation on $X \times X$ such that $\Aone$ and $\Athree$ hold,\footnote{Note that the order topology and $\Aone$ are redundant assumptions. The lemma follows immediately by continuity of $\wpp$ alone.} and suppose $(x_n)$, $(y_n)$ converge to $x$ and $y$ respectively, and $(w_n)$, $(z_n)$ converge to $w$ and $z$ respectively. If for every $n \in \mathbb{N}$ we have $(x_n,y_n) \wpp (w_n,z_n)$ then $(x,y) \wpp (w,z).$
\end{lemma}

\begin{proof}
Denote the set $A:= \{(x,y,w,z) \in \fourX : (x,y) \wpp (w,z)\}$ and pick a sequence of points with values in $A$, that is pick $(x_n,y_n, w_n,z_n) \in A^{\mathbb{N}}$ converging to $(x,y,w,z)$. By assumption, we have that $x_n \rightarrow x$, $y_n \rightarrow y$, $w_n \rightarrow w$, $z_n \rightarrow z$ and this is equivalent to $(x_n,y_n, w_n,z_n) \rightarrow (x,y,w,z)$. Indeed, a sequence in the product space $\fourX$ converges to $(x,y,w,z)$ if and only if it converges componentwise, $\ie$ $x_n \rightarrow x$, $y_n \rightarrow y$, $w_n \rightarrow w$, $z_n \rightarrow z$. We now prove this fact.

Assume $\fourxn \rightarrow \fourx$ in $\fourX$. Let $U_1, U_2, U_3, U_4$ be open sets containing $x,y,w,z$, respectively. Then $U_1 \times U_2 \times U_3 \times U_4$ is a basis element (hence, open) for the product topology containing $\fourx$. By definition of convergence, we can find $n_0$ such that for all $n \ge n_0$ we have $\fourxn \in U_1 \times U_2 \times U_3 \times U_4$. Thanks to the fact that projections are continuous functions, they preserve convergent sequences and so for all $n \ge n_0$ we have $x_n \in U_1$, $y_n \in U_2$, $w_n \in U_3$, $z_n \in U_4$, $\ie$ $x_n \rightarrow x$, $y_n \rightarrow y$, $w_n \rightarrow w$, $z_n \rightarrow z$.

Conversely, if $x_n \rightarrow x$, $y_n \rightarrow y$, $w_n \rightarrow w$, $z_n \rightarrow z$, let $U_{\star}$ be an open subset of $\fourX$ such that $\fourx \in U_{\star}$. By definition of product topology, we can find $U_1 \subseteq X$ open in $X$, $\dots$, $U_4 \subseteq X$ open in $X$ such that $x \in U_1$, $y \in U_2$, $w \in U_3$, $z \in U_4$. By convergence, we have that for all $i=1,2,3,4$ there exists $n_{k_{i}} \in \mathbb{N}$ such that for all $n \ge n_{k_{i}}$ we have $x_n \in U_1$, $y_n \in U_2$, $w_n \in U_3$, $z_n \in U_4$. Now pick $N:=\text{max} \{n_{k_{1}},n_{k_{2}},n_{k_{3}},n_{k_{4}}\} $ and for every $n \ge N$ we have $\fourxn \in U_1 \times U_2 \times U_3 \times U_4 \subseteq U_{\star}$. Hence, by definition of convergence, $\fourxn \rightarrow \fourx$.

Now we want to show $\fourx \in A$, with $A$ closed in the product topology. We now prove that every closed set in the product topology is sequentially closed.\footnote{Note that when $X$ is metrizable, a set $C \subseteq X$ is closed $\iff$ $C$ is sequentially closed.} This means we want to show that if we pick a sequence of points $\fourxn$ with values in $A \subseteq X$ that is converging to a point $\fourx \in X$, then $\fourx \in A$.
Pick a sequence $\fourxn$ with values in $A \subseteq X$ that is converging to a point $\fourx \in X$. Then, let $U_{\star}$ be any neighborhood of $\fourx$. By convergence, there exist an $n_0 \in \mathbb{N}$ such that for all $n \ge n_0$ we have $\fourxn \in U_{\star}$ and, in particular, $\fourxn \in U_{\star} \cap A$. Since $U_{\star}$ was an arbitrary but fixed neighborhood of $\fourx$, then $\fourx$ is in the closure of $A$, $\ie$ $\fourx \in \overline{A}$. But $A$ is closed, therefore $A=\overline{A}$, so $\fourx \in A$, hence $(x,y) \wpp (w,z)$. 
\end{proof}

The proof of Theorem \ref{thm1} in chapter 3, as in the original version of Shapley [\ref{shapley75}], relies on two very interesting lemmas. Similar propositions have been taken as axioms in environments that lack the topological assumptions on the set of alternatives $X$.

\begin{lemma}\label{1}
Let (w,z) be an element of $X \times X$. If $x', x'', y \in X$ are such that:
\begin{equation} (x',y) \succcurlyeq (w,z) \succcurlyeq (x'',y) \end{equation}
then there exists a unique, up to indifference, $x^{\star} \in X$ such that 
\begin{equation}\label {eq1}
(x^{\star},y) \sim (w,z) \end{equation} and $x' \succsim x^{\star} \succsim x''$.
\end{lemma}

\begin{proof}
Define $x_{0}:=\inf\{x\in X : (x,y) \succcurlyeq (w,z)\}$ and denote $A:=\{x\in X : (x,y) \succcurlyeq (w,z)\}$ this set. The set $A$ is nonempty as $x' \in A$, $A$ is bounded below by $x''$ as we have $(w,z) \succcurlyeq (x'',y) $ and, by transitivity and $\Aone$, we reach $x \succsim x''$ for every $x \in A$. Thus, $x_{0}$ is such that $x' \succsim x_{0}  \succsim x''$ and so $x_{0} \in X$ by Lemma \ref{wlc}. Analogously, we define $x^{0}:=\sup\{x\in X : (w,z) \succcurlyeq (x,y)\}$ and denote $B:=\{x\in X : (w,z) \succcurlyeq (x,y)\}$ this set. Then, $B$ is nonempty as $x'' \in B$, $B$ is bounded above by $x'$ as we have $(x',y) \succcurlyeq (w,z)$ and, by transitivity and $\Aone$, we reach $x' \succsim x$ for every $x \in B$. Thus, $x^{0}$ is such that $x' \succsim x^{0} \succsim x''$ and so $x^{0} \in X$ by Lemma \ref{wlc}.

By $\Athree$, the sets $A$ and $B$ are closed and so, by Lemma \ref{sup}, we have $x_{0} \in A$ and $x^{0} \in B$ so that $$(x_{0},y) \succcurlyeq (w,z) \succcurlyeq (x^{0},y)$$ 

\noindent
By transitivity and by $\Aone$ we have $x_{0} \succsim x^{0}$. 

Assume now by contradiction that $x_{0} \succ x^{0}$. By Lemma \ref{linearcontinuum} there exists $x^{\star} \in X$ such that $x^{0} \prec x^{\star} \prec x_{0}$. But then, comparing $x^{\star}$ with $(w,z)$, $(x^{\star}, y) \succcurlyeq (w,z)$ can hold only if $x_{0} \sim x^{\star} \succ x^{0}$, so $x_{0} \sim x^{\star}$ and therefore $x^{\star}$ should be the infimum of $A$, reaching a contradiction. Specular is the contradiction in the other case. Therefore, as there does not exist any $x^{\star} \in X$ such that $x^{0} \prec x^{\star} \prec x_{0}$, we must conclude that $x_{0} \sim x^{0}$. By transitivity and $\Aone$ we have $$ (x_0,y) \sim (w,z) \sim (x^0,y)$$ This proves the existence of $x^{\star} \in X$ for which (\ref{eq1}) holds.

Let $\overline{x} \in X$ be any other element of $X$ for which (\ref{eq1}) holds. By transitivity, $(x^{\star},y) \sim (\overline{x},y)$. By $\Aone$, we have $x^{\star} \sim \overline{x}$ and this completes the proof.
\end{proof}

\begin{lemma}\label{2}
Let $x, z \in X$ such that $x \succ z$. Then, there exists a unique, up to indifference, $y^{\star} \in X$ such that $$(x,y^{\star})\sim (y^{\star},z)$$ and $x \succ y^{\star} \succ z$.
\end{lemma}

\begin{proof}

Define $y^{0}$ to be the least upper bound of the set $C:=\{y\in X : (x,y) \succcurlyeq  (y,z)\}$. This set is nonempty as if we pick $y=z$ we have $(x,z) \succcurlyeq  (z,z)$ that by $\Aone$ is equivalent to $x \succsim  z$, that holds by assumption. $C$ is also bounded from above by $x$ as if we pick $y=x$ we have $(x,x) \succcurlyeq  (x,z)$ that by $\Aoneprime$ is equivalent to $z \succsim  x$, that, by completeness, contradicts the assumption of $x \succ  z$ showing that $x$ is an upper bound for $C$. Since $C$ is nonempty and bounded above by $x$, by Lemma \ref{wlc} we have $y^{0} \in  X$.

Similarly, by defining $y_{0}$ to be the greatest lower bound of the set $D:=\{y\in  X : (y,z) \succcurlyeq  (x,y)\}$. This set is nonempty as if we pick $y=x$ we have $(x,z) \succcurlyeq  (x,x)$ that by $\Aoneprime$ is if and only if $x \succsim  z$, that holds by assumption. This set is also bounded from below by $z$ as if we pick $y=z$ we have $(z,z) \succcurlyeq  (x,z)$ that by $\Aone$ is if and only if $z \succsim  x$, that, by completeness, contradicts the assumption of $x \succ  z$ showing that $z$ is a lower bound for $D$. Since $D$ is nonempty and bounded below by $z$, by Lemma \ref{wlc} we have $y_{0} \in X$.

By $\Athree$ the sets $C$ and $D$ are closed, so by Lemma \ref{sup} we have $y^{0} \in  C$ and $y_{0} \in D$, that is \begin{equation}\label{eq7}
(x,y^0) \succcurlyeq  (y^0,z) \text{ and } (y_0,z) \succcurlyeq  (x,y_0)
\end{equation}
We show now that $y^{0} \succsim y_{0}$. Suppose, by contradiction, $y_{0} \succ  y^{0}$. By Lemma \ref{linearcontinuum} there exists $y^{\star} \in  X$ such that $y_{0} \succ  y^{\star} \succ  y^{0}$. Then, by definition of $y_{0}$ we have $(y^{\star},z) \prec (x, y^{\star})$, while by the definition of $y^{0}$ we have $(x,y^{\star}) \prec (y^{\star}, z)$. This contradiction shows that $y^{0} \succsim y_{0}$. By $\Aone$ this is equivalent to 
\begin{equation}\label{eq8}
(y^0,z) \succcurlyeq  (y_0,z) \text{ for all } z \in  X.
\end{equation}
By $\Aoneprime$ it is also equivalent to 
\begin{equation}\label{eq9}
(x,y_0) \succcurlyeq  (x,y^0) \text{ for all } x \in  X. 
\end{equation}
Putting together equation \ref{eq7} with equations \ref{eq8} and \ref{eq9}, we reach the loop $$(y^0,z) \succcurlyeq  (y_0,z) \succcurlyeq  (x,y_0) \succcurlyeq  (x,y^0) \succcurlyeq  (y^0,z).$$ By transitivity, we have $(y^0,z) \sim  (y_0,z)$ and $(x, y_0) \sim  (x,y^0)$. By $\Aone$, we conclude that $y^0 \sim  y_{0}$.
\end{proof}

We conclude proving that from $\Aone$, $\Atwo$ and $\Athree$ we can derive $\Aoneprime$. 


\begin{lemma} \label{lemmaaprime}
Let $X$ be a connected subset of a topological space. If $\succsim$ is complete and transitive, $\wpp$ is complete, transitive, satisfies $\Athree$ and $\Atwo$, and jointly $\succsim$ and $\wpp$ satisfy $\Aone$, then $\Aoneprime$ holds, that is, for all $x,y,z \in X$ we have $x \succsim y$ if and only if $(z,y) \wpp (z,x)$.
\end{lemma}

\begin{proof}
By contradiction, suppose $\Aoneprime$ fails. Then, there exist $x,y,z \in X$ such that $(z,y) \wpp (z,x)$ and $x \prec y$. We consider two cases: $y \succ z$ and $ y \precsim z$. 

If $y \succ z$ then, being $(z,y) \wpp (z,x)$ by assumption, we have $$(x,x) \sim (y,y) \succ (z,y) \wpp (z,x)$$ by $\Atwo$ and $\Aone$, respectively. We apply Lemma \ref{1} to find a $w \in X$ such that $$(w,x) \sim (z,y) \text{  and  } x \succsim w \succsim z.$$ Being $y \succ x$, we have $$(z,z) \sim (y,y) \succ (x,y) \sim (w,z) \succsim (z,z)$$ by $\Atwo$, $\Aone$, $\Atwo$, $\Aone$, respectively. This implies a contradiction in the case $y \succ z$.

Assume now $y \precsim z$. Being $y \precsim z$ and $x \prec y$, by transitivity we have $x \prec z$. We can proceed as in the previous case, interchanging the roles of $x$ and $y$ and reversing all the inequalities.
\end{proof}

\section{The theorem}

\renewcommand{\appendixtocname}{Appendix}
\renewcommand{\appendixpagename}{Appendix}

We can now state and prove Shapley's theorem in our general version.

\begin{theorem}\label{thm1}
Let $X$ be a connected and separable subset of a topological space. If $\succsim$ is complete and transitive, $\wpp$ is complete, transitive, satisfies $\Atwo$ and $\Athree$, and jointly $\succsim$ and $\wpp$ satisfy $\Aone$, then the pair $(\succsim, \wpp)$ can be represented by a continuous measurable utility function u: $ X \rightarrow \mathbb{R}$, that is, for each pair $x,y \in X$,
\begin{equation}\label{eq1.10}
 x \succsim  y \iff u(x) \ge u(y) 
\end{equation}
and for each quadruple $x,y,z,w \in X$,
\begin{equation}\label{eq1.11}
 (x,y) \succcurlyeq  (z,w) \iff u(x) - u(y) \ge u(z) - u(w) .
\end{equation}
Moreover, $u$ is unique up to positive affine transformations.
\end{theorem}

\begin{proof}

We first prove the result when $\succsim$ is antisymmetric. In view of Lemma \ref{linearcontinuum}, throughout the proof we will consider suprema and infima of subsets of $X$.

Suppose $X$ is not a singleton, otherwise the result is trivially true. Let $a_0, a_1 \in X$ be two distinct elements of $X$ such that, without loss of generality, $a_1 \succ  a_0$. 

Assign $u(a_0)=0$ and $u(a_1)=1$. Now we want to show that $u$ has a unique extension on $X$ which is a measurable utility function for $(\succsim, \wpp)$.
To ease notation, denote $$\bold1:=(a_1,a_0) \text{  ,  } \bold0:=(a_0,a_0) \text{  ,  } \bold{-1}:=(a_0,a_1).$$

\noindent
Clearly, $\bold1, \bold0 ,\bold{-1} \in X \times X$ and, by $\Aone$ and $\Aoneprime$, $\bold1 \succ \bold0 \succ \bold{-1}$.
Then, by $\Atwo$ we have $(x,x) \sim  \bold0$ for every $x \in X$. Moreover, for every $y \in X$ we have either:
$$(i)\hspace{0.15cm} \text{There exists a unique} \hspace{0.15cm} T_1(y) \in X  \hspace{0.15cm} \text{such that} \hspace{0.15cm} (T_1(y),y) \sim  \bold1$$ 
or
$$(ii) \hspace{0.15cm} \bold1 \succ  (x,y) \hspace{0.08cm} \text{ for all } x \in X\hspace{6cm} $$
Indeed, if $(ii)$ fails, there exists $x' \in X$ such that $(x',y) \succcurlyeq  \bold1$. 
Since $(x',y) \succcurlyeq  \bold1 \succcurlyeq  \bold0 \sim  (y,y)$, by Lemma \ref{1} there exists an element $T_1(y) \in X$ such that $(T_1(y),y) \sim \bold1$. By $\Aone$ and antisymmetry of $\succsim$ , $(T_1(y),y) \sim (y',y)$ implies $T_1(y) = y'$, so $  T_1(y)$ is unique. 

In addition, note that $y \prec T_1(y)$. Indeed, $(y,y) \sim \bold0 \prec \bold1 \sim (T_1(y),y)$, and so $\Aone$ implies $y \prec T_1(y)$. In a similar way as before, for every $y \in X$ we have either:
$$(i.bis)\hspace{0.15cm} \text{There exists a unique} \hspace{0.15cm} T_{-1}(y) \in X  \hspace{0.15cm} \text{such that} \hspace{0.15cm} (T_{-1}(y),y) \sim  \bold{-1}$$ 
or
$$(ii.bis) \hspace{0.15cm} \bold{-1} \prec  (x,y) \hspace{0.08cm} \text{ for all } x \in X\hspace{6.35cm} $$
Indeed, if $(ii.bis)$ fails, there exists $x' \in X$ such that $(x',y) \preccurlyeq  \bold{-1}$. 
Since $(x',y) \preccurlyeq  \bold{-1} \preccurlyeq  \bold0 \sim  (y,y)$, by Lemma \ref{1} there exists an element $T_{-1}(y) \in X$ such that $(T_1(y),y) \sim \bold{-1}$. By $\Aone$ and antisymmetry of $\succsim$ , $(T_{-1}(y),y) \sim (y',y)$ implies $T_{-1}(y) = y'$, so $  T_{-1}(y)$ is unique. 

In addition, note that $T_{-1}(y) \prec y$. Indeed, $ (T_{-1}(y),y) \sim \bold{-1} \prec \bold0 \sim (y,y)$, and so $\Aone$ implies $T_{-1}(y) \prec y$.

Now define $a_2 := T_1(a_1)$ if $(i)$ holds for $y=a_1$, $\ie$ if there exists a unique $T_1(a_1) \in X$ such that $(T_1(a_1),a_1) \sim  \bold1$. Similarly, set $a_3:= T_1(a_2)$ if $(i)$ holds for $y=a_2$, and continue in this way till (if ever) occurs $y=a_{n}$ for which $(ii)$ holds, $\ie$ $\bold1 \succ  (x,a_{n})$ for every $x \in X$. Analogously, we define $a_{-1}:= T_{-1}(a_0)$ if $(i.bis)$ holds for $y=a_0$, set $a_{-2}:= T_{-1}(a_{-1})$ if $(i.bis)$ holds for $y=a_{-1}$, and continue in this way till (if ever) occurs $y=a_{-n}$ for which $(ii.bis)$ holds.

Now define $\mathcal{A} := \{\dots, a_{-2}, a_{-1}, a_0, a_1, a_2, \dots \}$, with $$ \dots \prec a_{-2} \prec a_{-1} \prec a_0 \prec a_1 \prec a_2 \prec \dots$$

\noindent
The set $\mathcal{A}$ can be finite or infinite in either direction. If we consider now a sequence that from an index set $P_{a} \subseteq \mathbb{Z}$ maps to $\mathcal{A}$, we define the following function $a : P_{a} \subseteq \mathbb{Z} \rightarrow \mathcal{A}$. 

Now we start to extend $u$ to $\mathcal{A}$. Define the following: $$u(a_{p}) = p \hspace{0.4cm} \text{for every } p \in P_{a}.$$ Clearly, we have (\ref{eq1.10}), $\ie$ $x \succsim y$ if and only if $u(x) \ge u(y)$ for every $x,y$ that are images of the sequence $a$, so (\ref{eq1.10}) holds on $\mathcal{A}$.

Now we show that (\ref{eq1.11}) holds whenever $x,y,z,w \in \mathcal{A} \subset X$, say $x=a_p, y=a_q, z=a_{p-d}$ where $p,q,p-d \in P_a$. Without loss of generality, assume $d \ge 0$. We first prove the ``equality" case of (\ref{eq1.11}), that is \begin{equation}\label{eq1.22}
(x,y) \sim  (z,w) \iff u(x) - u(y) = u(z) - u(w)
\end{equation}

\noindent
By construction we have $$(a_{p},a_{p-1}) \sim  \bold1 \sim  (a_{q},a_{q-1})$$ so, by transitivity and $\Atwo$, we have: $$(a_{p},a_{q}) \sim  (a_{p-1},a_{q-1})$$ 

\noindent
Iterating this procedure finitely many times we reach: \begin{equation}\label{eq1.23}
(x,y) = (a_{p},a_{q}) \sim  (a_{p-d},a_{q-d}) = (z,a_{q-d})
\end{equation}

By transitivity, $(z, a_{q-d}) \sim (z,w)$ and so, by $\Aone$ $a_{q-d}=w$, so that $u(a_{q-d}) = u(w)$. By definition of $u$ we can write $$u(x) - u(y) = u(a_{p}) - u(a_{q}) = p-q = u(a_{p-d}) - u(a_{q-d}) = u(z) - u(w)$$thus proving (\ref{eq1.22}). Next we prove \begin{equation}\label{eq1.24}
(x,y) \succ  (z,w) \iff u(x) - u(y) > u(z) - u(w)
\end{equation}

\noindent
By transitivity, $(z, a_{q-d}) \succ (z,w)$ and so, by $\Aoneprime$, $w \succ a_{q-d}$, so that $u(w) > u(a_{q-d})$. By definition of $u$, from (\ref{eq1.23}) we can write $$u(x) - u(y) = u(a_{p}) - u(a_{q}) = p-q = u(a_{p-d}) - u(a_{q-d}) > u(z) - u(w)$$ thus proving (\ref{eq1.24}).

Summing up, both (\ref{eq1.10}) and (\ref{eq1.11}) hold on the terms of the set $\mathcal{A}$. Using Lemma \ref{2}, now we want to extend $u$ to the points of $X$ that lie between terms of the set $\mathcal{A}$.
Set $b_{0}:=a_0$ and since $a_1 \succ a_0$, by Lemma \ref{2} there exists $b_1 \in X$, with $a_1 \succ b_1 \succ a_0$, such that $$(a_1,b_1) \sim  (b_1,a_0)$$ Now build the set $\mathcal{B} := \{\dots, b_{-2}, b_{-1}, b_0, b_1, b_2, \dots \}$, with $$ \dots \prec b_{-2} \prec b_{-1} \prec b_0 \prec b_1 \prec b_2 \prec \dots$$ based on $b_0, b_1$, in the same way we constructed $\mathcal{A}$ from $a_0, a_1$. Also here, we can define a sequence that from an index set $P_{b} \subseteq \mathbb{Z}$ maps to $\mathcal{B}$, that is, we define the following function $b : P_{b} \subseteq \mathbb{Z} \rightarrow \mathcal{B}$.

By construction we have $$(b_2,b_1) \sim  (b_1,b_0)$$ Together with $(a_1,b_1) \sim  (b_1,a_0)$, by transitivity we have $(b_2,b_1) \sim  (a_1,b_1)$. By $\Aone$, $b_2 = a_1$. Analogously, one can verify that 
\begin{equation}\label{eq1.25}
b_{2p} = a_{p} \text{  for every  } p\in P_a \end{equation}

So, the terms of the set $\mathcal{B}$ lie between the terms of the set $\mathcal{A}$, $\ie$ the set $\mathcal{B}$ refines $\mathcal{A}$ and we can write 
\begin{equation}\label{eq1.26}
\mathcal{A} \subseteq \mathcal{B}
\end{equation}

Denote now $c_0:=b_{0}=a_0$ and we let $c_1 \in X$ be that element provided by Lemma \ref{2} such that $(b_1,c_1) \sim  (c_1,b_0)$. In the same way we constructed $\mathcal{B}$ from $\mathcal{A}$, we can construct, from $\mathcal{B}$, a third set $\mathcal{C}:= \{\dots, c_{-2}, c_{-1}, c_0, c_1, c_2, \dots \}$, with $$ \dots \prec c_{-2} \prec c_{-1} \prec c_0 \prec c_1 \prec c_2 \prec \dots$$ based on $c_0, c_1$. We can see that $$c_{2p} = b_{p} \text{  for every  } p\in P_c$$ where $P_c \subseteq \mathbb{Z}$ is the collection of indexes of the sequence $c : P_{c} \subseteq \mathbb{Z} \rightarrow \mathcal{C}$.

The set $\mathcal{C}$ refines $\mathcal{B}$
\begin{equation}\label{eq1.27}
\mathcal{B} \subseteq \mathcal{C}
\end{equation}

We keep iterating this process, constructing sets that refine one another and, for ease of notation, we denote them in the following way: 

$$\mathcal{A}_0 := \mathcal{A} \hspace{1cm}\text{and}  \hspace{1cm} a^{0}_p:= a_{p} \in  \mathcal{A}_0$$
$$\mathcal{A}_1 :=\mathcal{B}  \hspace{1cm}\text{and}  \hspace{1cm} a^{1}_p:= b_{p} \in  \mathcal{A}_1$$
$$\mathcal{A}_2 :=\mathcal{C}  \hspace{1cm}\text{and}  \hspace{1cm} a^{2}_p:= c_{p} \in  \mathcal{A}_2$$
$$\cdots$$

\noindent
These sets generalize the inclusions (\ref{eq1.26}) and (\ref{eq1.27}) as follows:
\begin{equation}\label{eq1.28}
\mathcal{A}_0 \subseteq \mathcal{A}_1 \subseteq \mathcal{A}_2 \subseteq \dots \subseteq \mathcal{A}_n \subseteq \dots
\end{equation}

\noindent
So, in general, $a^{n}_p$ for $p \neq 1$ is obtained from the construction of $(i)$ and $(ii)$, applied to the points $a_0, a^{n}_1$. The term $a^{n}_1$, for $n>0$, is the ``midpoint" between $a^{n-1}_1$ and $a_0$, that exists by Lemma \ref{2}. By iterating the construction of (\ref{eq1.25}), we have that $$\frac{p}{2^{n}} = \frac{q}{2^{m}} \Longrightarrow a^{n}_p = a^{m}_q $$
\noindent
In the spirit of (\ref{eq1.28}), we extend $u$ to all points in $\mathcal{A}_{\infty} := \bigcup_{n=1}^{\infty}\mathcal{A}_{n}$ by: $$u(a^{n}_p)= \frac{p}{2^{n}} \hspace{0.5cm} \text{ for all }  a^{n}_p \in \mathcal{A}_{n}$$

Relations (\ref{eq1.10}) and (\ref{eq1.11}) hold in this extended domain: given $x,y,z,w \in \bigcup_{n=1}^{\infty}\mathcal{A}_{n}$, just take $n$ large enough so that they become, up to indifference, terms of the set $\mathcal{A}_{n}$ and proceed in the same exact way as we did for the set $\mathcal{A}_0$.

To complete the construction of $u$ we only remain to show $\mathcal{A}_{\infty}$ is dense in $X$, that is $\overline{\mathcal{A}_{\infty}} = X$. We first show that none of the sets $\mathcal{A}_{n}$ has, for its sequences of points $a^n$, a point of accumulation in $X$. Indeed, fix $n$ and suppose by contradiction that $a^n_{p_{k}}$ converges monotonically to $a^{\star} \in X$, where, without loss of generality, we assume $a^n_{p_{k}} \uparrow a^{\star}$ with $a^{\star} \in X$, $\ie$ $(p_{k})$ is an increasing sequence of integers. Denote $\bold1_{n} := (a^{n}_1,a_0)$ and we have, for every $k \in \mathbb{N}$, $$(a^{n}_{1 + p_{k}},a^{n}_{p_{k}}) \succcurlyeq  \bold1_{n} \succ  \bold0  $$
By Lemma \ref{producttopology}, we have $(a^{\star},a^{\star}) \succcurlyeq  \bold1_{n}$. So, by transitivity, we reach $(a^{\star},a^{\star})   \succ  \bold0$, a contradiction. We conclude that, fixed $n$, none of the sequences $a^{n}$ with values in $\mathcal{A}_n$ has a limit point in $X$.

To prove $\overline{\mathcal{A}_{\infty}} = X$, the implication $\overline{\mathcal{A}_{\infty}} \subseteq X$ is trivial by construction. Now we want to show $\overline{\mathcal{A}_{\infty}} \supseteq X$, that is all the elements of $X$ belong to the closure of $\mathcal{A}_{\infty}$ as well. Fix $x \in X$ such that, without loss of generality, $x \succsim  a_0$. For $n \ge 1$, define $y_{n}:=\sup\{y \in \mathcal{A}_{n} : x \succsim  y\}$. Note that $a_0 \in \{y \in \mathcal{A}_{n} : x \succsim  y\}$, so this set is nonempty and we can write $x \succsim  y_{n} \succsim  a_{0}$. By Lemma \ref{wlc}, $y_{n} \in X$. Note further that, as shown before, $\mathcal{A}_{n}$ cannot have accumulation points in $X$ so, as long as $y_{n} \in X$, it follows $y_{n}$ cannot be an accumulation point of $\mathcal{A}_{n}$. So, $y_n$ must belong to $\mathcal{A}_{n}$ and we denote $y_{n}:= a^n_{p_{n}}$. As a result, we have: 
\begin{equation}\label{eq1.29}
a^n_{p_{n}-k} \precsim x \prec a^n_{p_{n}+k} \text{  for every  } k >0
\end{equation}
We also have that
\begin{equation}\label{eq1.30}
\bold1_{n} \succ  (x,y_{n}) \end{equation}
Indeed, if (\ref{eq1.30}) were not true, then $(x,a^n_{p_{n}}) \wpp \bold1_n$. We consider two cases: $a^n_{1+p_{n}} \succ x$ or $a^n_{1+p_{n}}\precsim x$. If $a^n_{1+p_{n}} \succ x$, then, thanks to $\Aone$, we reach the following contradiction: \begin{equation}
\bold1_{n} \sim (a^n_{1+p_{n}}, a^n_{p_{n}}) \succ (x, a^n_{p_{n}}) \wpp \bold1_n
\end{equation}
So $a^n_{1+p_{n}}\precsim x$, but this contradicts (\ref{eq1.29}), that is, it contradicts $y_{n}$ to be the supremum. Thus, (\ref{eq1.30}) holds. In particular, by $\Aone$ and $\Atwo$, we can write $(x, y_n) \wpp (y_n,y_n) \sim \bold0$, leading to
\begin{equation} \label{eq1.31}
\bold1_{n} \succ  (x,y_{n}) \wpp \bold0
\end{equation}

Now, when $n \rightarrow \infty$, as the sets $\mathcal{A}_{n+1} \supseteq \mathcal{A}_{n} \supseteq \mathcal{A}_{n-1} \dots$ are nested one into the other by (\ref{eq1.28}), we can write, for every $n \ge 1$, $y_n \precsim y_{n+1} \precsim x$. Thus, the points $y_{n}$ form a non-decreasing sequence that is bounded from above by $x$. Call $y^{\star}$ the limit of this sequence, that is well-defined by Lemma \ref{linearcontinuum}. Since $a_0 \precsim y^{\star} \precsim x$, by Lemma \ref{wlc} it follows that $y^{\star} \in X$. In particular, by Lemma \ref{metric} we have $y^{\star} \in \overline{\mathcal{A}_{\infty}}$, because, for every fixed $n \ge 1$, $y_n$ is a term of the sets $\mathcal{A}_{n}$, and so $(y_n) \in \mathcal{A}_{\infty}^{\mathbb{N}}$.

As to the $\bold1_{n}$ terms, for $n$ fixed, we see that $$\bold1_n \sim  (a^n_2, a^n_1) \sim (a^{n-1}_1, a^n_1)$$ We also have that, for every $n \ge 1$, $a_0 \precsim a^{n+1}_1 \precsim a^n_1$.

Thus, the points $a^{n}_{1}$ form, for $n \rightarrow \infty$, a non-increasing sequence that is bounded from below by $a_0$. Call $a^{\star}$ the limit of this sequence, that is well-defined by Lemma \ref{linearcontinuum}. Since $a_0 \precsim   a^{\star} \precsim  a_1$, by Lemma \ref{wlc} we have $a^{\star} \in X$.

Consider now $(a^{n-1}_1, a^{n}_1)$ and $(x, y_n)$. By Lemma \ref{producttopology} and from (\ref{eq1.31}) it follows that $$(a^{\star},a^{\star}) \wpp (x,y^{\star}) \wpp \bold0$$
Since, by $\Atwo$, $(a^{\star},a^{\star}) \sim \bold0$, by transitivity $ (x,y^{\star}) \sim \bold0$, so that $x \sim y^{\star}$, $\ie$ $x = y^{\star}$ as $\succsim$ is antisymmetric.

Since $x$ was arbitrarily chosen in $X$ and $y^{\star} \in \overline{\mathcal{A}_{\infty}}$, we can conclude $x \in \overline{\mathcal{A}_{\infty}}$, so that $\overline{\mathcal{A}_{\infty}} = X$. Therefore, we can extend $u$ by continuity to the whole set $X$ by setting $$u(x) = \lim_{n \rightarrow \infty}u(x_n)$$ if $(x_n) \in \mathcal{A}_{\infty}^{\mathbb{N}}$ converges monotonically to $x$. Note that $u: X \rightarrow \R$ is well-defined. Indeed, to prove it is well-posed we show that if $x_n$ and $y_n$ are two sequences that converge to $x$, then $\lim_{n \rightarrow \infty}u(x_n)=\lim_{n \rightarrow \infty}u(y_n)$. This follows easily by continuity of $u$.\footnote{Recall that in every topological space $X$ continuity implies sequential continuity. The converse holds if $X$ is first-countable.} In light of Lemma \ref{metric}, it is easy to see that $u$ satisfies (\ref{eq1.10}) and (\ref{eq1.11}).

As to uniqueness, observe that any other $\overline{u}$ that satisfies (\ref{eq1.10}) and (\ref{eq1.11}) can be normalized so that $\overline{u}(a_0) = 0$ and $\overline{u}(a_1)=1$. So, $\overline{u}$ must agree on $u$ at each step of the constructive procedure for $u$ just seen. Indeed, for a given $\overline{u}: X \rightarrow \R$, define the following positive affine transformation $f: \text{Im}(\overline{u}) \rightarrow \R$ such that $$f(x):= \frac{x - \overline{u}(a_0)}{\overline{u}(a_1) - \overline{u}(a_0)}$$ It is immediate to see that, for the equivalent utility function $\widehat{u} := f \circ \overline{u}$, we have $\widehat{u}(a_0) = 0$ and $\widehat{u}(a_1) = 1$.

Summing up, we proved Theorem \ref{thm1} if $\succsim$ is antisymmetric. Now we drop this assumption. Let $\equivclass{X}$ be the quotient space with respect to the equivalence relation $\sim$. The set $\{x \in X : x \sim y\}$ is a closed set in $X$ by Lemma \ref{t1}, so $(\equivclass{X}, \thicktilde{\succsim})$ is a totally ordered connected and separable subset of a topological space, where $\thicktilde{\succsim}$ is the total order induced by the weak order $\succsim$.\footnote{That is, $\thicktilde{\succsim}:= \text{}\equivclass{\succsim}$ $\subseteq$ $\equivclass{X}$ $\times$ $\equivclass{X}$.} Therefore, the orders $\succsim$ and $\wpp$ induce orders $\thicktilde{\succsim}$ and $\thicktilde{\wpp}$ on the quotient set $\equivclass{X}$, by setting, for all $[x],[y]  \in \equivclass{X}$ $$[x] \text{  } \thicktilde{\succsim} \text{  } [y] \iff x \succsim y$$ and, for all $[x],[y],[z],[w] \in \equivclass{X}$ $$([x], [y]) \text{  } \thicktilde{\wpp} \text{  } ([z],[w]) \iff (x,y) \wpp (z,w)$$ It is routine to show that the orders $\thicktilde{\succsim}$ over $\equivclass{X}$ and $\thicktilde{\wpp}$ over $\equivclass{X} \times \equivclass{X}$ inherit the same properties of $\succsim$ and $\wpp$ used in the theorem. So, by what has been proved so far, there exists $\thicktilde{u}:  \equivclass{X} \rightarrow \R$ that satisfies (\ref{eq1.10}) and (\ref{eq1.11}) for $(\thicktilde{\succsim}, \thicktilde{\wpp})$. Let $\pi: X \rightarrow \equivclass{X}$ be the quotient map. Then, the function $u: X \rightarrow \mathbb{R}$ defined as $u= \thicktilde{u} \circ \pi$ is a well-defined measurable utility function, $\ie$ it is easily seen to satisfy (\ref{eq1.10}) and (\ref{eq1.11}) for $(\succsim, \wpp)$.

To conclude, we show that $u$ satisfies (\ref{eq1.10}) and (\ref{eq1.11}). If $x \sim y$ then $[x]=[y]$ and, by the theorem we have just proved, $\thicktilde{u}([x])=\thicktilde{u}([y])$, which is $(\thicktilde{u} \circ \pi )(x)=(\thicktilde{u} \circ \pi )(y)$, and so $u(x)=u(y)$. If $x \succ y$, then $[x] \succ  [y]$, which implies $\thicktilde{u}([x])>\thicktilde{u}([y])$, which is $(\thicktilde{u} \circ \pi )(x)>(\thicktilde{u} \circ \pi )(y)$, and so $u(x)>u(y)$.

Conversely, assume $u(x) \ge u(y)$ and suppose by contradiction $x \nsucceq y$ that, by completeness, is $y \succ x$. If $u(x) = u(y)$ then $\thicktilde{u}([x])=\thicktilde{u}([y]) \iff [x]=[y] \iff x \sim y$, a contradiction. If $u(x) > u(y)$ then $\thicktilde{u}([x])>\thicktilde{u}([y]) \iff [x]>[y] \iff x \succ y$, a contradiction. Hence, (\ref{eq1.10}) holds for $u$.

By definition, we have that $([x],[y]) \succcurlyeq  ([z],[w]) \iff (x,y) \succcurlyeq (z,w)$, for all [x],[y],[z],[w] $\in \equivclass{X}$. So, we can write $(x,y) \succcurlyeq (z,w) \iff ([x],[y]) \succcurlyeq  ([z],[w]) \iff \thicktilde{u}([x]) - \thicktilde{u}([y]) \ge \thicktilde{u}([z]) - \thicktilde{u}([w]) \iff u(x)- u(y) \ge u(z) - u(w).$ Hence, also (\ref{eq1.11}) holds for $u$.

This completes the proof of Theorem \ref{thm1}.
\end{proof}

\noindent
Graphically, we can build the following diagram to represent our construction.

\begin{center} \Huge
\begin{tikzcd}[column sep=2cm, row sep=2cm,font= \huge , label= \huge ]
  X \arrow[r, twoheadrightarrow, "\pi"] \arrow[dr,"u"', red]
& \equivclass{X} \arrow[d, "\thicktilde{u}"]\\
& \mathbb{R}
\end{tikzcd}
\end{center}

\end{document}